\documentclass[11pt,a4paper]{article}
\usepackage{url}
\usepackage{colortbl,colortab}
\usepackage{booktabs}
\usepackage{caption}
\usepackage{rotating}
\usepackage{multirow}
\usepackage{amsmath,amssymb,amsthm}
\usepackage{enumerate}
\usepackage{subfigure}
\usepackage{epsfig}
\usepackage{color,booktabs}
\usepackage{lscape}
\usepackage{verbatim}
\usepackage{algorithm}
\usepackage{dcolumn}
\usepackage{threeparttable}
\usepackage{longtable}
\usepackage{algorithmicx,algorithm,algpseudocode}
\usepackage{frcursive}
\usepackage{geometry}
\usepackage{pbox}
\usepackage{pst-all}
\usepackage{pstricks-add}
\usepackage{thmtools}
\usepackage{thm-restate}
\allowdisplaybreaks
\usepackage{rotating}
\usepackage{siunitx}
\usepackage[affil-it]{authblk}
\usepackage{makecell}
\usepackage{indentfirst}
\usepackage{tikz}
\algblockdefx[For]{ForEach}{EndForEach}[1]{\textbf{for each} #1 \textbf{do}}{\textbf{end for}}

\usetikzlibrary{positioning}
\usetikzlibrary{shapes}
\usetikzlibrary{arrows.meta}
\tikzstyle{block} = [rounded corners, inner sep=5pt, very thick, rectangle, draw, font=\fontsize{9}{9}\selectfont]
\tikzstyle{specialise} = [very thick, {-Stealth[length=2mm, width=2mm]}]

\usepackage[normalem]{ulem}

\newtheorem{theorem}{Theorem}

\newtheorem{definition}{Definition}
\newtheorem{example}{Example}

\newcommand{\rank}[2]{\ensuremath{r^{\scriptscriptstyle c}_{#2}(#1)}}
\newcommand{\rankp}[2]{\ensuremath{r^{\scriptscriptstyle f}_{#2}(#1)}}
\newcommand{\rankT}[2]{\ensuremath{r^{\scriptscriptstyle d}_{#2}(#1)}}
\newcommand{\rankpT}[2]{\ensuremath{r^{\scriptscriptstyle h}_{#2}(#1)}}
\newcommand{\SIJ}[2]{\ensuremath{F^{\scriptscriptstyle \leq}_{#2}(#1)}}
\newcommand{\TJI}[2]{\ensuremath{C^{\scriptscriptstyle \leq}_{#2}(#1)}}
\newcommand{\SIJT}[2]{\ensuremath{H^{\scriptscriptstyle \leq}_{#2}(#1)}}
\newcommand{\TJIT}[2]{\ensuremath{D^{\scriptscriptstyle \leq}_{#2}(#1)}}
\newcommand{\GI}[1]{\ensuremath{g^{\scriptscriptstyle c}{(#1)}}}
\newcommand{\GPJ}[1]{\ensuremath{g^{\scriptscriptstyle f}{(#1)}}}
\newcommand{\GIT}[1]{\ensuremath{g^{\scriptscriptstyle d}{(#1)}}}
\newcommand{\GPJT}[1]{\ensuremath{g^{\scriptscriptstyle h}{(#1)}}}
\newcommand{\FKI}[2]{\ensuremath{F^{\scriptscriptstyle =}_{#1}(#2)}}
\newcommand{\CKJ}[2]{\ensuremath{C^{\scriptscriptstyle =}_{#1}(#2)}}
\newcommand{\HKI}[2]{\ensuremath{H^{\scriptscriptstyle =}_{#1}(#2)}}
\newcommand{\DKJ}[2]{\ensuremath{D^{\scriptscriptstyle =}_{#1}(#2)}}
\newcommand{\YIK}[2]{\ensuremath{y^{\scriptscriptstyle c}_{#1 #2}}}
\newcommand{\YPJK}[2]{\ensuremath{y^{\scriptscriptstyle f}_{#1 #2}}}
\newcommand{\YIKT}[2]{\ensuremath{y^{\scriptscriptstyle d}_{#1 #2}}}
\newcommand{\YPJKT}[2]{\ensuremath{y^{\scriptscriptstyle h}_{#1 #2}}}

\def\rev#1{{#1}}

\begin{document}

\pagestyle{plain}

\title{\bf Mathematical models for stable matching problems with ties and incomplete lists}

\author{Maxence Delorme$^{(1)}$, Sergio Garc\'{i}a$^{(1)}$, Jacek Gondzio$^{(1)}$, \\ Joerg Kalcsics$^{(1)}$,
David Manlove$^{(2)}$, William Pettersson$^{(2)}$}

\affil{
\small $(1)$ School of Mathematics, University of Edinburgh, United Kingdom\\
\small $(2)$ School of Computing Science, University of Glasgow, United Kingdom\\
\vspace*{1ex}
{\small \em Corresponding author \texttt{maxence.delorme@ed.ac.uk}, phone +44 0131 650 5870}}
\date{ }
\maketitle
\noindent
\vspace*{-5ex}
\begin{abstract}

We present new integer linear programming (ILP) models for $\mathcal{NP}$-hard optimisation problems in instances of the Stable Marriage problem with Ties and Incomplete lists (SMTI) and its many-to-one generalisation, the Hospitals / Residents problem with Ties (HRT).
These models can be used to efficiently solve these optimisation problems when applied to (i) instances derived from real-world applications, and (ii) larger instances that are randomly-generated.
In the case of SMTI, we consider instances arising from the pairing of children with adoptive families, where preferences are obtained from a quality measure of each possible pairing of child to family.
In this case we seek a maximum weight stable matching.
We present new algorithms for preprocessing instances of SMTI with ties on both sides, as well as new ILP models.
Algorithms based on existing state-of-the-art models only solve~6~of our~22~real-world instances within an hour per instance, and our new models \rev{incorporating dummy variables and constraint merging, together with preprocessing and a warm start}, solve all~22~instances within a mean runtime of a minute.
For HRT, we consider instances derived from the problem of assigning junior doctors to foundation posts in Scottish hospitals.
Here we seek a maximum size stable matching.
\rev{We show how to extend our models for SMTI to HRT and reduce the average running time for real-world HRT instances
by two orders of magnitude.}
We also show that our models outperform \rev{by a wide margin all known} state-of-the-art models on larger randomly-generated instances of SMTI and HRT.
\end{abstract}


\section{Introduction}\label{Prob}
\subsection{Background}

In a stable matching problem, we are given a set of agents, each of whom ranks all the others in strict order of preference, indicating their level of desire to be matched to each other.
A~solution of the problem is a pairing of all agents such that no two agents form a {\em blocking pair}, i.e., a pair that are not currently matched together, but would prefer to be matched to each other rather than to their currently assigned partners.

Without any other constraints, this problem is known as the Stable Roommates (SR) problem \cite{GS62,GI89}, and the objective is to partition the $n$ agents into $n/2$ pairs (e.g., doubles in a tennis tournament) such that no blocking pair exists.

The Stable Marriage problem (SM) is a \rev{bipartite restriction} of SR, where the agents are split into equal-sized sets of men and women, and it is assumed that men only find women acceptable and vice versa.
This problem was first introduced by Gale and Shapley \cite{GS62}, who also gave a linear-time algorithm for finding a stable matching.

It is not always desirable, or even possible, to have every agent express a preference over all other agents.
In the Stable Marriage problem with Incomplete lists (SMI), agents can identify potential partners as being unacceptable, meaning that they would rather be unmatched than matched to such agents, and a slight modification of the Gale-Shapley algorithm will find a stable matching in linear time \cite[Section 1.4.2]{GI89}.
It turns out that all stable matchings in a given instance of SMI have the same size \cite{GS85}.

\rev{In many applications it is not realistic to expect that agents have sufficient information to enable them to rank their acceptable potential partners in strict order of preference.
In reality, preference lists may include ties, where a tie indicates a set of agents that are equally desirable.
This gives rise to another variant of SM known as the Stable Marriage problem with Ties (SMT) \cite{Irv94}.
It is known that resolving indifference by employing tie-breaking is not a good strategy, since it over-constrains the problem \cite{EE08}.
Instead, three levels of stability \cite{Irv94} have been defined in the SMT case, where ties are retained, that vary according to whether agents will agree to swap between choices they find equally acceptable.}
Under the weakest of these three definitions, which we assume in this paper, a stable matching can always be found by arbitrarily breaking the ties, resulting in an instance of SM.

If both ties and incomplete lists are introduced we obtain the Stable Marriage problem with Ties and Incomplete lists, or SMTI \cite{MIIMM02}.
In an instance of SMTI, stable matchings do not necessarily have the same size, and MAX-SMTI, the problem of finding a stable matching of maximum size, is ${\cal NP}$-hard \cite{MIIMM02}.

The Stable Roommates problem with Globally Ranked Pairs (SR-GRP) \cite{ALMO08,ABEOMP09} is a variant of the Stable Roommates problem involving ties and incomplete lists in which each pair of compatible agents $\{p,q\}$ has a weight $w(\{p,q\})$ assigned to their potential pairing, and the preference lists of each agent can be derived from these weights in the obvious manner: given two compatible pairs $\{p,q\}$ and $\{p,r\}$, $p$ prefers $q$ to $r$ if and only if $w(\{p,q\})>w(\{p,r\})$.
This problem can be restricted to give the Stable Marriage problem with Ties, Incomplete lists, and Globally Ranked Pairs (SMTI-GRP) by splitting the agents into two sets as per the Stable Marriage problem.

In this work, we study one application of SMTI-GRP involving the pairing of children with adoptive families as coordinated by the British charity Coram\footnote{Coram | Better chances for children since 1739, \url{https://www.coram.org.uk}}.
Social workers \rev{determine} a weight to \rev{be assigned to} each child--family pair $(c,f)$, as a predicted measure of the suitability of placing~$c$ with~$f$, giving an instance of SMTI-GRP.
\rev{Currently Coram is using a clearing house system which pairs children and families at suitable specified intervals.
Similar to the case for kidney exchange programmes \cite{RSU04}, this allows for a more efficient pairing of children and families, at the cost of a slightly increased delay between entering the system and being paired.}
\rev{In such a system Coram has decided that} the goal \rev{should be} to find a stable matching that pairs as many children as possible and/or has maximum overall weight\footnote{The child--family pairings in a computed stable matching are treated merely as suggestions that will be investigated further by social workers for suitability before any actual assignments are made.}.
Moreover, Coram would like to ensure that the computed matching is viable in the long term.
To this end, a lower bound, or threshold, on suitable weights is used to create refined instances of SMTI-GRP where child--family pairs with weights below the threshold are not allowed to be matched together.
\rev{However, attempts to determine appropriate threshold values, as well as good weighting functions and suitable intervals between matching runs, have been hampered by the lack of tractable algorithms for finding maximum weight stable matchings for such instances.}
\rev{Indeed, in} the SMTI-GRP setting, ${\cal NP}$-hardness holds for each of the problems of finding a maximum size stable matching \cite{ALMO08} and a maximum weight stable matching~\cite{DMS17}.

Whilst SMTI is a one-to-one matching problem, in some applications one set of agents can be matched with more than one partner.
The Hospitals / Residents (HR) problem \cite{GS62,Man15} is a many-to-one extension of SMI that models the assignment of intending junior doctors (residents) to hospitals.
Each doctor is to be assigned to at most one hospital, whilst each hospital may be assigned multiple doctors up to some given capacity.
HR can be generalised to include ties in the preference lists, leading to the Hospitals / Residents with Ties (HRT), the many-to-one generalisation of SMTI.
HRT has many applications: it models, for example, the assignment of medical graduates to Scottish hospitals as part of the Scottish Foundation Allocation Scheme (SFAS), which ran between 1999 and 2012.
Since then, the UK has amalgamated all such schemes into the UK Foundation Programme, which handles the assignment of almost 8000~doctors to approximately 7000 positions across 20 Foundation Schools, each of which consists of multiple hospitals \cite{UKFP}.
In this setting a key aim is to find a stable matching of maximum size, which is an ${\cal NP}$-hard problem in view of the ${\cal NP}$-hardness of MAX-SMTI.

\rev{An overview of the differences between problems discussed in the paper is given in Table~\ref{table:matching_types}.
The relationships between these problems are demonstrated in Figure~\ref{diagram:relationships}.
In the diagram, an arrow from problem A to problem B indicates that problem B is a special case of problem A.
For example, SMTI-SYM is the special case of SMTI-GRP in which preferences are symmetric.}

\begin{table}[hbt]
\begin{center}
\caption{\rev{Summary of matching problems}}\label{table:matching_types}
\begin{tabular}{@{}cccccc@{}}
\toprule
Variant & Bipartite & Incompatible pairs & Ties & Weights & Capacity\\
\cmidrule(r){1-1} \cmidrule(r){2-2} \cmidrule(r){3-3} \cmidrule(r){4-4} \cmidrule(r){5-5} \cmidrule(r){6-6}
SR & No & No & No & No & No\\
SR-GRP & No & Yes & Yes & Yes & No \\
SM & Yes & No & No & No & No \\
SMI & Yes & Yes & No & No & No \\
SMT & Yes & No & Yes & No & No \\
SMTI & Yes & Yes & Yes & No & No \\
SMTI-GRP & Yes & Yes & Yes & Yes & No \\
SMTI-SYM & Yes & Yes & Yes & Yes & No \\
HRT & Yes & Yes & Yes & No & Yes \\
\bottomrule
\end{tabular}
\end{center}
\end{table}

\begin{figure}[hbt]
\begin{center}
\caption{\rev{Relationships between matching problems}}\label{diagram:relationships}
\begin{tikzpicture}
\node[block] (sm) at (0,0) {SM};
\node[block, left=1cm of sm] (sr) {SR};
\node[block, above right=1cm and 1cm of sm] (smi) {SMI};
\node[block, below right=1cm and 1cm of sm] (smt) {SMT};
\node[block, right=3cm of sm] (smti) {SMTI};
\node[block, right=1cm of smti] (hrt) {HRT};
\node[block, below=1cm of smti] (smti-grp) {SMTI-GRP};
\node[block, below=1cm of sr] (sr-grp) {SR-GRP};
\node[block, right=1cm of smti-grp] (smti-sym) {SMTI-SYM};
\draw[specialise] (sr) -- (sm);
\draw[specialise] (smi) -- (sm);
\draw[specialise] (smt) -- (sm);
\draw[specialise] (smti) -- (smi);
\draw[specialise] (smti) -- (smt);
\draw[specialise] (hrt) -- (smti);
\draw[specialise] (smti) -- (smti-grp);
\draw[specialise] (sr) -- (sr-grp);
\draw[specialise] (smti-grp) -- (smti-sym);
\end{tikzpicture}
\end{center}
\end{figure}
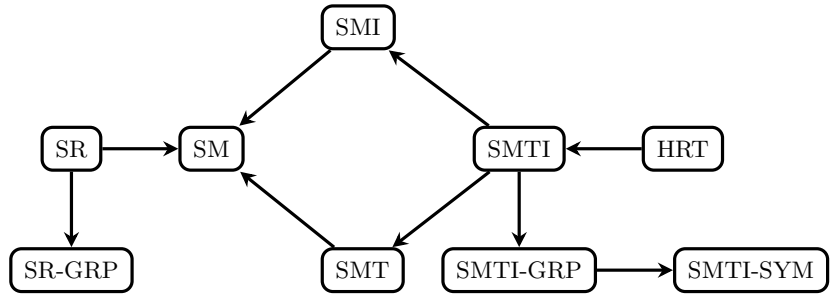
\subsection{Our contribution}

In this paper we have developed several new techniques that improve the performance of ILP models for instances of both SMTI and HRT.
Our first contribution is to present two algorithms for preprocessing instances of SMTI with ties on both sides.
Without such preprocessing, only 6~of~22~real-world instances from Coram could be solved within an hour per instance using state-of-the-art models from the literature.
Our new preprocessing significantly improves this, finding solutions to 21 of the 22 instances in an average of 434~seconds.
We also present new ILP models for SMTI and HRT.
These use dummy variables to reduce the number of non-zero entries in their corresponding constraint matrices, which vastly increases the sparsity of the constraint matrix at the cost of additional variables.
Further, we formulate different sets of constraints to model stability, including the use of redundant constraints to improve the continuous relaxations of our models.
We test each of these individually, and these improvements together allow us to find solutions to all real-world instances in a mean runtime of less than 60 seconds.
Turning to randomly-generated instances, the new models also solve all 30 random instances of SMTI that we generated with $\num{50000}$ agents on either side and preference lists of length 5 on one side, while existing state-of-the-art models could only solve 20.
We extend our new ILP models to HRT, where we show a reduction in the mean runtime on existing real-world instances of HRT from SFAS, decreasing the average runtime from 144 seconds to only~3~seconds.
We also generate~90~\rev{random} instances that mimic the UK Foundation Programme (with about 7500~doctors and positions).
Existing models solve 66 of these, while our new models solve 81.

\subsection{Related work}\label{Liter}
MAX-SMTI is known to be ${\cal NP}$-hard even if each tie occurs at the end of some agent's preference list, ties occur on one side only and each tie is of length two \cite{MIIMM02}.
The special case of MAX-SMTI that asks if an instance of SMTI has a stable matching that matches every man and woman is also ${\cal NP}$-complete~\cite{MIIMM02}, and this result holds even when preference lists have lengths of at most~3 and ties occur on one side only~\cite{IMO09}.

MAX-SMTI is also not approximable within a factor of 21/19~\cite{HIMY07} unless $\cal P=\cal NP$, even if preferences on one side are strictly ordered, and on the other side are either strictly ordered or a tie of length two.
The best currently-known performance guarantee is 3/2, achieved first in non-linear running time~\cite{McD09} and later improved to linear time~\cite{Kir13,Pal14}, although better guarantees can be achieved in certain restrictions~\cite{IM16}.
Kir\'aly \cite{Kir13}  shows how to extend his 3/2-approximation algorithm for MAX-SMTI to HRT.

The Stable Marriage problem with Ties, Incomplete lists and Symmetric preferences (SMTI-SYM) is a restriction of SMTI-GRP such that (i) for each man--woman pair $(u,v)$, the rank of $v$ in $u$'s list, i.e., the integer $k$ such that $v$ belongs to the $k$th tie in $u$'s list, is equal to the rank of $u$ in $v$'s list, and (ii) the weight of $(u,v)$ is precisely this integer $k$.  Finding a maximum size stable matching in an instance of SMTI-SYM is ${\cal NP}$-hard, and therefore the same result holds for SMTI-GRP~\cite{ALMO08}.
Given an instance of SMTI-GRP, if the goal is to find a matching that maximises the total weight rather than the total size, this problem is ${\cal NP}$-hard also~\cite{DMS17}.

Linear programming models for SM and SMI have been long studied, and stable matchings correspond exactly to extreme points of the solution polytopes of such models~\cite{GI89,VV89}.
These formulations have been extended to give ILP models for finding maximum size stable matchings in instances of SMTI and HRT~\cite{Kwa15,KM14}.
ILP models have also been given for a common extension of HR that allows doctors to apply as couples, typically so that both members can be matched to hospitals that are geographically close~\cite{ABM16,DPB15,Hin15,MMT17,McB15}.
Other techniques in the field include constraint programming, which has been applied to SM and its variants~\cite{GIMPS01,GP02b,MOPU05,OMa07}, and the use of SAT models and SAT solvers~\cite{DPB15,GP02b}.

Diebold and Bichler \cite{DB17} performed a thorough experimental study of eight algorithms for~HRT, giving a comparison of these algorithms when applied to real-world HRT instances derived from a course allocation system at the Technical University of Munich.
These datasets ranged in size from 18-733 students (the ``doctors'') and 3-43 courses (the ``hospitals'').
The authors measured a number of attributes of the algorithms, including the sizes of the computed stable matchings.
The \rev{methods} that they considered included three exact algorithms for MAX-HRT based on the ILP model presented in \cite{KM14}.

\rev{Slaugh et al.\ \cite{SAKU16} described improvements they had made to the mechanism for matching children to adoptive families as utilized by the Pennsylvania Adoption Exchange.
The process is semi-decentralized in that up to ten match attempts are made against families when each child arrives.
By contrast, the more centralized process adopted by Coram involves a pool of children and families building up over time, leading to the use of a matching algorithm for the resulting two-sided matching problem.}

For more details on the diverse variants of stable matching problems, we direct the reader to \cite{Man13} and for an economic overview of these problems we recommend \cite{RS90}.

\subsection{Layout of the paper}
The rest of the paper is organised as follows.
Section \ref{Prob2} defines the problems that are studied in this paper, and we introduce and discuss existing models for these in Section \ref{Form}.
This is followed by a theorem and two algorithms for preprocessing instances of SMTI \rev{in order to reduce instance sizes,} in Section \ref{Pre}.
Section \ref{Dum} introduces our first new model, which reduces the number of non-zero elements in the constraint matrix through dummy variables.
Further models are presented in Section \ref{Sta} with new stability constraints.
We demonstrate our new models and improvements experimentally in Section \ref{Comp} and we provide some conclusions in Section \ref{Con}.

\section{Problem definitions} \label{Prob2}
In this section we give formal definitions of the three key problems that we consider in this paper.
\subsection{Stable Marriage with Ties and Incomplete Lists}
\label{sec:SMTIdef}
An instance $I$ of the \emph{Stable Marriage problem with Ties and Incomplete lists} (\emph{SMTI}) comprises a set $C$ of $n_1$ children and a set $F$ of $n_2$ families, where each child (respectively family) ranks a subset of the families (respectively children) in order of preference, possibly with ties.
We say that a child $c\in C$ finds a family $f\in F$ \emph{acceptable} if $f$ belongs to $c$'s preference list, and we define acceptability for a family in a similar way.
We assume that preference lists are \emph{consistent}, that is, given a child--family pair $(c,f)\in C\times F$, $c$ finds $f$ acceptable if and only if $f$ finds $c$ acceptable. If $c$ does find $f$ acceptable then we call $(c,f)$ an \emph{acceptable pair}.

A \emph{matching} $M$ in $I$ is a subset of acceptable pairs such that, for each agent $a\in C\cup F$, $a$ appears in at most one pair in $M$.
If $a$ appears in a pair of $M$, we say that $a$ is \emph{matched}, otherwise $a$ is \emph{unmatched}.
In the former case, $M(a)$ denotes $a$'s \emph{partner} in $M$, that is, if $(c,f)\in M$, then $M(c)=f$ and $M(f)=c$.
We now define stability, which is the key condition that must be satisfied by a matching in $I$.
\begin{definition}
\label{def:stab_smti}
Let $I$ be an instance of SMTI and let $M$ be a matching in $I$.
A child-family pair $(c,f)\in (C\times F)\backslash M$ is a \emph{blocking pair} of $M$, or \emph{blocks} $M$, if
\begin{enumerate}
\item $(c,f)$ is an acceptable pair,
\item either $c$ is unmatched in $M$ or $c$ prefers $f$ to $M(c)$, and
\item either $f$ is unmatched in $M$ or $f$ prefers $c$ to $M(f)$.
\end{enumerate}
$M$ is said to be \emph{stable} if it admits no blocking pair.
\end{definition}
In SMTI, the goal is to find an arbitrary stable matching.
We denote the problem of finding a maximum size stable matching, given an instance of SMTI, by MAX-SMTI.

\subsection{Globally Ranked Pairs}
\label{sec:smti-grp}
An instance $I$ of the \emph{Stable Marriage problem with Ties, Incomplete lists and Globally-Ranked Pairs} (\emph{SMTI-GRP}) comprises a set $C$ of $n_1$ children, a set $F$ of $n_2$, a subset $X\subseteq C\times F$ of \emph{acceptable} child--family pairs, and a weight function $w: X \longrightarrow \mathbb R$.

The set of acceptable pairs and the weight function are used to define the SMTI instance~$I'$ corresponding to $I$ as follows: for any two acceptable pairs $(c,f)\in X$ and $(c,f')\in X$, $c$ prefers $f$ to $f'$ if $w(c,f) > w(c,f')$, and $c$ is indifferent between $f$ and $f'$ if $w(c,f) = w(c,f')$.
Preference lists of families are constructed in a similar fashion.
A stable matching in $I$ can then be defined by applying Definition \ref{def:stab_smti} to $I'$.

Given a matching $M$ in $I$, the \emph{weight} of $M$, denoted by $w(M)$, is defined to be $\sum_{(c,f)\in M} w(c,f)$.
The problem of finding a stable matching of maximum size is called MAX-SMTI-GRP, and the problem of finding a stable matching of maximum weight is called MAX-WT-SMTI-GRP.

Given an instance $I$ of MAX-WT-SMTI-GRP, we can construct a refined instance $I'$ of MAX-WT-SMTI-GRP from $I$ by setting a threshold value $t$ with the effect that the acceptable pairs in $I'$ are precisely the acceptable pairs in $I$ which have weight at least $t$.
The effect of imposing different threshold values on $I$ is of interest to Coram.

\begin{example}
Our first example demonstrates how different threshold values create instances of SMTI-GRP with differently sized maximum size stable matchings.
Let $C=\{c_1,c_2,c_3\}$ be a set of children, $F=\{f_1,f_2,f_3\}$ be a set of families, and let the weight function $w$ be defined by the following table:
\begin{center}
\begin{tabular}{|c|c|c|c|}\hline
      & $f_1$ & $f_2$ & $f_3$ \\ \hline
$c_1$ & 95 & 85 & 80 \\ \hline
$c_2$ & 95 & 80 & 80 \\ \hline
$c_3$ & 80 & 45 & 75 \\ \hline
\end{tabular}
\end{center}
\vspace*{2ex}
By taking $t=0$ we obtain an instance of SMTI-GRP in which all pairs are acceptable. In this instance, $M_1=\{(c_1,f_2),(c_2,f_1),(c_3,f_3)\}$ is the unique maximum weight stable matching and its weight is 255.
However, if we take $t=80$ and construct an instance of SMTI-GRP, then the only acceptable pair that involves $c_3$ is $(c_3,f_1)$ and no stable matching can involve $c_3$.
The unique maximum weight stable matching is then $M_2=\{(c_1,f_2),(c_2,f_1)\}$, which has a weight of~180.
\end{example}
\begin{example}
Our second instance of SMTI-GRP is intended to show that a maximum weight stable matching may be smaller in size than a maximum size stable matching.
Let $C=\{c_1,c_2,c_3,c_4\}$, $F=\{f_1,f_2,f_3,f_4\}$, \[X=\{(c_i,f_i) : 1\leq i\leq 4\}\cup \{(c_{i+1},f_i) : 1\leq i\leq 3\},\] and the weight function $w$ be given by the following table:
\begin{center}
\begin{tabular}{|c|c|c|c|c|}\hline
      & $f_1$ & $f_2$ & $f_3$ & $f_4$ \\ \hline
$c_1$ & 1 & -- & -- & -- \\ \hline
$c_2$ & 4 & 4 & -- & -- \\ \hline
$c_3$ & -- & 3 & 4 & -- \\ \hline
$c_4$ & -- & -- & 4 & 1 \\ \hline
\end{tabular}
\end{center}
Let $M_1=\{(c_i,f_i) : 1\leq i\leq 4\}$ and $M_2=\{(c_{i+1},f_i) : 1\leq i\leq 3\}$.
It is easy to verify that $M_1$ and $M_2$ are both stable matchings.
However $w(M_1)=10$ and $w(M_2)=11$, whereas $|M_1|=4$ and $|M_2|=3$.
\end{example}

\subsection{Hospitals / Residents with Ties}
An instance $I$ of the \emph{Hospitals / Residents problem with Ties} (\emph{HRT}) comprises a set $D$ of $n_1$~resident doctors and a set $H$ of $n_2$ hospitals.
Each doctor (respectively hospital) ranks a subset of the hospitals (respectively doctors) in order of preference, possibly with ties.
Additionally, each hospital $h$ has a \emph{capacity} $c_h\in \mathbb Z^+$, meaning that $h$ can be assigned at most $c_h$ doctors, while each doctor is assigned to at most one hospital.
The definitions of the terms \emph{consistent} and \emph{acceptable} are analogous to the SMTI case.

A {\em matching} $M$ in $I$ is a subset of acceptable pairs such that each doctor appears in at most one pair, and each hospital $h\in H$ appears in at most $c_h$ pairs.
Given a doctor $d\in D$, the terms \emph{matched} and \emph{unmatched}, and the notation $M(d)$, are defined as in the SMTI case.
Given a hospital $h\in H$, we let $M(h)=\{d\in D : (d,h)\in M\}$.
We say that $h$ is \emph{full} or \emph{undersubscribed} in $M$ if $|M(h)|=c_h$ or $|M(h)|<c_h$, respectively.
We next define stability by extending Definition \ref{def:stab_smti} to the HRT case.
\begin{definition}
\label{def:stab_hrt}
Let $I$ be an instance of HRT and let $M$ be a matching in $I$.
A doctor--hospital pair $(d,h)\in (D \times H)\setminus M$ is a \emph{blocking pair} of $M$, or \emph{blocks} $M$, if
\begin{enumerate}
\item $(d, h)$ is an acceptable pair,
\item either $d$ is unmatched in $M$ or $d$ prefers $h$ to $M(d)$, and
\item either $h$ is undersubscribed in $M$ or $h$ prefers $d$ to some member of $M(h)$.
\end{enumerate}
$M$ is said to be \emph{stable} if it admits no blocking pair.
\end{definition}

As in the SMTI case, the problem of finding a maximum size stable matching, given an instance of HRT, is denoted MAX-HRT.

While the definition for HRT does allow an arbitrary number of preferences to be expressed by any doctor, in reality doctors' preference lists are often short: for example in the SFAS application until 2009, every doctor's list was of length 6.

\section{Existing formulations}\label{Form}
The first mathematical models for SM were proposed in the late 1980s independently by Gusfield and Irving \cite{GI89} and by Vande Vate \cite{VV89}.
Rothblum \cite{Rot92} extended Vande Vate's model to SMI.
In the following, we show how to extend Rothblum's model to formulate both MAX-SMTI and MAX-HRT, as was done previously by Kwanashie and Manlove \cite{KM14}.
These existing models for MAX-SMTI and MAX-HRT are described here as they will be extended in later sections.

\subsection{Mathematical model for MAX-SMTI}
Based on our Coram application, we will adopt the terminology from that context when presenting models for MAX-SMTI.
When reasoning about models, we will use $i$ and $j$ to represent a child and family, rather than $c$ and $f$, respectively, as $i$ and $j$ are by convention more typically used as subscript variables.
Let us consider the following notation:
\begin{itemize}
\item $F(i)$ is the set of families acceptable for child $i$ $(i=1,\dots,n_1)$.
\item $C(j)$ is the set of children acceptable for family $j$ $(j=1,\dots,n_2)$.
\item $\rank{i}{j}$ is the rank of family $j$ for child $i$, defined as the integer $k$ such that $j$ belongs to the $k$th most-preferred tie in $i$'s list $(i=1,\dots,n_1,\ j \in F(i))$.
    The smaller the value of~$\rank{i}{j}$, the better family $j$ is ranked for child $i$.
\item $\rankp{j}{i}$ is the rank of child $i$ for family $j$, defined as the integer $k$ such that $i$ belongs to the $k$th most-preferred tie in $j$'s list $(j=1,\dots,n_2,\ i \in C(j))$. The smaller the value of~$\rankp{j}{i}$, the better child $i$ is ranked for family $j$.
\item $\SIJ{i}{j}$ is the set of families that child $i$ ranks at the same level or better than family~$j$, that is, $\SIJ{i}{j}=\{j'\in F : \rank{i}{j'} \leq \rank{i}{j}\}$ $(i=1,\dots,n_1,\ j \in F(i))$.
\item $\TJI{j}{i}$ is the set of children that family $j$ ranks at the same level or better than child~$i$, that is, $\TJI{j}{i}=\{i'\in C : \rankp{j}{i'} \leq \rankp{j}{i}\}$ $(j=1,\dots,n_2,\ i \in C(j))$.
\end{itemize}
    By introducing binary decision variables $x_{ij}$ that take value $1$ if child~$i$ is matched with family~$j$, and $0$ otherwise $(i=1,\dots,n_1,\ \rev{j \in F(i)})$, MAX-SMTI can be modelled as follows:
\begin{alignat}{3}
&\max \quad  && \sum_{i = 1}^{n_1} \sum_{j \in F(i)} x_{ij} \label{eq:C1} \\
&\mbox{~s.t.}  && \sum_{j \in F(i)} x_{ij} \leq 1, \quad && i = 1,\dots,n_1, \label{eq:C2} \\
& && \sum_{i \in C(j)} x_{ij} \leq 1, \quad && j = 1,\dots,n_2, \label{eq:C3} \\
& && 1-\sum_{q \in \SIJ{i}{j}} x_{iq} \leq \sum_{p \in \TJI{j}{i}} x_{pj}, \quad && i = 1,\dots,n_1,\ j \in F(i), \label{eq:C4} \\
& && x_{ij} \in \{0,1\}, \quad&& i = 1,\dots,n_1,\ j \in F(i) \label{eq:C5} .
\end{alignat}

The objective function \eqref{eq:C1} maximises the number of children assigned.
If instead, one wants to maximise the score of the children assigned (as in MAX-WT-SMTI-GRP), it is enough to use $\sum_{i = 1}^{n_1} \sum_{j \in F(i)} w_{ij}x_{ij}$ in the objective function.
Constraints \eqref{eq:C2} ensure that each child is matched with at most one family and constraints \eqref{eq:C3} impose that each family is matched with at most one child.
Finally, constraints \eqref{eq:C4} enforce stability by ruling out the existence of any blocking pair.
More specifically, they ensure that if child $i$ is not matched with family $j$ or any other family they rank at the same level or better than $j$ (i.e., $\sum_{q \in \SIJ{i}{j}} x_{iq} = 0$), then family $j$ is matched with a child it ranks at the same level or better than $i$ (i.e., $\sum_{p \in \TJI{j}{i}} x_{pj} \geq 1$).

\subsection{Mathematical model for MAX-HRT}
An adaptation of model \eqref{eq:C1}-\eqref{eq:C5} for MAX-HRT was proposed in \cite{KM14}.
It uses the same notation that was used for MAX-SMTI except that:
\begin{itemize}
\item The term ``family'' is replaced by ``hospital'' and $F(i)$, $\rankp{j}{i}$, and  \SIJ{i}{j} are changed into~$H(i)$, $\rankpT{j}{i}$, and $\SIJT{i}{j}$, respectively.
\item The term  ``child'' is replaced by ``doctor'' and  $C(j)$, $\rank{i}{j}$, and $\TJI{j}{i}$ are changed to $D(j)$, $\rankT{i}{j}$, and $\TJIT{j}{i}$, respectively.
\item The capacity of hospital $j$ $(j=1,\dots,n_2)$ is denoted by $c_j$.
\end{itemize}

By introducing binary decision variables $x_{ij}$ that take value $1$ if doctor $i$ is assigned to~hospital $j$, and $0$ otherwise $(i=1,\dots,n_1,\, \rev{j \in H(i)})$, MAX-HRT can be modelled as follows:
\begin{alignat}{3}
&\max \quad  && \sum_{i = 1}^{n_1} \sum_{j \in H(i)} x_{ij} \label{eq:mod1} \\
&\mbox{~s.t.}  && \sum_{j \in H(i)} x_{ij} \leq 1, \quad && i = 1,\dots,n_1, \label{eq:mod2} \\
& && \sum_{i \in D(j)} x_{ij} \leq c_j, \quad && j = 1,\dots,n_2, \label{eq:mod3} \\
& && c_j\left(1-\sum_{q \in \SIJT{i}{j}} x_{iq}\right) \leq \sum_{p \in \TJIT{j}{i}} x_{pj}, \quad && i = 1,\dots,n_1,\ j \in H(i), \label{eq:mod4} \\
& && x_{ij} \in \{0,1\}, \quad&& i = 1,\dots,n_1,\ j \in H(i) \label{eq:mod5} .
\end{alignat}

While the meaning of the objective function and constraints \eqref{eq:mod2} remains the same, constraints \eqref{eq:mod3} ensure now that each hospital does not exceed its capacity. Constraints~\eqref{eq:mod4} are the adaptation of the stability constraints \eqref{eq:C4} when capacity is considered.
More specifically, they ensure that if doctor $i$ was not assigned to hospital $j$ or any other hospital they rank at the same level or higher than $j$ (i.e., $\sum_{q \in \SIJT{i}{j}} x_{iq} = 0$), then hospital$~j$ has filled its capacity with doctors it ranks at the same level or higher than $i$ (i.e., $\sum_{p \in \TJIT{j}{i}} x_{pj} \geq c_j$).

\subsection{Discussion on the models}
Although the model for SM was proposed almost thirty years ago, the computational behaviour of its extension to MAX-SMTI and MAX-WT-SMTI-GRP (i.e., in one-to-one instances specifically) has never been studied, to the best of our knowledge.
However, we mention that our direct implementation of \eqref{eq:C1}-\eqref{eq:C5} on real-world MAX-WT-SMTI-GRP instances involving $500$ children, $\num{1000}$ families, and a large list of preferences cannot be solved by state-of-the-art solvers within hours.
Indeed, the model becomes too difficult as it requires up to $\num{500000}$ stability constraints, each of them including $|\SIJ{i}{j}| + |\TJI{j}{i}|$ nonzero elements (i.e., up to $\num{1500}$).

Regarding MAX-HRT, computational experiments with \eqref{eq:mod1}-\eqref{eq:mod5} applied to real-world and randomly generated instances have been carried out previously \cite{DB17,Kwa15,KM14}.  Kwanashie \cite{Kwa15} observed a significant increase in terms of average running time when the number of doctors goes above 400.  As our objective is to solve instances of the magnitude of the UK Foundation Programme application (involving almost $\num{8000}$ doctors and $500$ hospitals), the model in its current form is not suitable.

In the next sections, we introduce various techniques aimed at reducing the size of the two models and strengthening their continuous relaxation.


\section{Preprocessing SMTI with ties on both sides}\label{Pre}

It is quite common in combinatorial optimisation to use some simple analysis to fix the optimal value of a subset of variables and, thus, reduce the problem size.
\rev{This is particularly useful  for stable matching problems as one variable, one stability constraint, and up to $n_1 + n_2$ non-zero elements are associated with each acceptable pair.}
Two procedures, ``Hospitals-offer'' and ``Residents-apply', have been proposed for removing acceptable pairs that cannot be part of any stable matching for HRT when ties only occur in hospitals' preference lists \cite{IM09}.

When ties can belong to the preference lists of both sets of agents, a reduction technique is known for the special case of SMTI in which preference lists on one side are of length at most two \cite{IMO09}.
However the aforementioned preprocessing algorithms are not applicable to arbitrary instances of SMTI.
In this section we introduce a new sufficient condition to find a set of acceptable pairs that cannot be part of any stable matching for SMTI.
We then propose two greedy algorithms to detect such pairs which can then be removed from the instance without affecting any stable matching.
Our technique is based on the following result.
\begin{theorem}\label{thm:preprocessing}
Let $I$ be an instance of SMTI.
Given a child $i$ and a set of families ${\cal F}$ such that for every family $j\in\mathcal{F}$, $(i,j)$ is an acceptable pair, let ${\cal C}$ be the set of children that at least one family in ${\cal F}$ ranks at the same level or better than $i$, i.e., $\mathcal{C} = \bigcup_{j\in\mathcal{F}} \{\TJI{j}{i}\}$. If $|{\cal F}| \geq |{\cal C}|$, then in any stable matching $M$, child~$i$ will be matched with a family $j'$ such that $\rank{i}{j'} \leq \max_{j \in {\cal F}}\{\rank{i}{j}\}$.
\end{theorem}

\begin{proof}
Assume for a contradiction that $M$ is a stable matching in $I$ in which child $i$ is matched with a family $j'$ with $\rank{i}{j'} > \max_{j \in \mathcal{F}} \rank{i}{j}$ or is unmatched.  Since $|\mathcal{F}| \geq |\mathcal{C}| > |\mathcal{C} \setminus \{i\}|$, at least one family $j''\in \mathcal{F}$ must be matched with some child~$i' \not\in \mathcal{C}$ or be unmatched.   Then either $i$ is unmatched or prefers $j''$ to $j'$, and either $j''$ is unmatched or prefers $i$ to $i'$.
In all cases $(i,j'')$ blocks $M$, which is a contradiction.
\end{proof}

There is no obvious efficient way to find, for each child, the set $\mathcal{F}$ that removes the largest number of acceptable pairs from an instance of SMTI.
Instead we present two polynomial-time algorithms to find sets that allow a significant number of acceptable pairs to be removed.
Algorithm \ref{alg-sfr}, ``first-rank-family'', considers the first rank of children for each family $j$, i.e.,~the children that $j$ thinks are the most desirable.
Algorithm \ref{alg-ecp}, ``full-child-preferences'', completely analyses the preference lists of the children to find reductions.
Note that each of these algorithms can also be applied to the preferences of the other set of agents by symmetry to obtain ``first-rank-child" and ``full-family-preferences", and that they may each be applied iteratively until no more reductions are possible.

\begin{algorithm}[hbt]
\caption{{\tt first-rank-family}}\label{alg-sfr}
\begin{algorithmic}[1]
\State \textbf{Input:} An instance of SMTI with children $C$ and families $F$
\State \textbf{Output:} A set $\mathcal{R}$ containing pairs $(i,j')$ that cannot be part of any stable matching
\ForEach{${\cal C} \in \mathcal{P}(C)$}  \Comment{for each subset of children in the powerset ${\cal P}(C)$}
\State ${\cal M}_{{\cal C}} \leftarrow \emptyset $
\EndForEach
\State $\mathcal{R} \leftarrow \emptyset$
\ForEach{family $j \in F$}
\State ${\cal C} \leftarrow \{i\in C(j) : \rankp{j}{i}=1\}$ \Comment{the set of children family $j$ considers equally best}
\State ${\cal M}_{{\cal C}} \leftarrow {\cal M}_{{\cal C}} \cup \{j\}$
\EndForEach
\ForEach{${\cal C} \in \mathcal{P}(C)$}
\State $\cal{F} \leftarrow \mathcal{M}_{{\cal C}}$
\If{$|{\cal F}| \geq |{\cal C}|$}
\ForEach{$i \in {\cal C}$}
\ForEach{$j' \in F(i)$ with $\rank{i}{j'} > \max_{j \in {\cal F}}\{\rank{i}{j}\}$}
\State $\mathcal{R} \leftarrow \mathcal{R} \cup \{(i,j')\}$
\EndForEach
\EndForEach
\EndIf
\EndForEach
\State \textbf{return }{$\mathcal{R}$}
\end{algorithmic}
\end{algorithm}

After initialisation (lines 3--6), Algorithm \ref{alg-sfr} considers each family $j$ in turn, determining the set of children $\mathcal{C}$ that family $j$ ranks as (equally) most desirable (line 8) and storing this fact (line 9).
Once this has been recorded, the algorithm searches through all these stored sets (line~11) to find sets of children ${\cal C}$ and the set of families ${\cal F}$ which all consider the set ${\cal C}$ as their (equally) best choice.
If the set of families ${\cal F}$ is at least as big as the set of children ${\cal C}$ (line 13) then for each child $i \in {\cal C}$ and each family $j' \in F(i)$ ranked worse than the worst family in ${\cal F}$, we add the pair $(i, j')$ to our reduction set $\cal R$ (lines~14--16).

As written, Algorithm~\ref{alg-sfr} requires $\mathcal{O}(2^{n_1}n_1 n_2)$ operations, as we must iterate over each possible subset of children (in both lines 3 and 11).
However, if we only explicitly store the subsets~$\cal C$ and $\cal M_{\cal C}$ generated by lines 7-10, we will obtain at most $n_2$ subsets $\cal C$ and at most $n_2$~subsets~$\cal M_{\cal C}$.
To only store these specific subsets, we need to quickly look up whether such a set $\mathcal{C}$ exists, and create it if it does not, before adding a family~$j$ to $\mathcal{M}_\mathcal{C}$.
A hashmap is a suitable data structure for carrying out these operations, and will reduce the overall complexity to $\mathcal{O}(n_1 n_2^2)$.

Algorithm \ref{alg-ecp} incrementally builds up the sets $\mathcal{F}$ and $\mathcal{C}$ for each child $i$.
To build $\mathcal{F}$, we simply add each family $j$ from the preference list of $i$ in order from most preferable to least (lines 6--7), considering agents within ties in increasing indicial order.
At each step, when we have added $j$, we then add to $\mathcal{C}$ all children that $j$ finds at least as preferable as $i$ (line 8).
By construction these satisfy Theorem \ref{thm:preprocessing}.
Thus, if $\mathcal{F}$ is large enough compared to $\mathcal{C}$, we add to our reduction all the pairs $(i, j')$ where  $j' \in F(i)$ are the families ranked worse than the worst family in ${\cal F}$ (lines~9--11).
Algorithm~\ref{alg-ecp} requires $\mathcal{O}(n_1n_2(n_1+n_2))$ steps as the outer (respectively middle and inner) \textbf{for each} loop is executed $\mathcal{O}(n_1)$ (respectively $\mathcal{O}(n_2)$ and $\mathcal{O}(n_2)$) times, and line 8 requires $O(n_1)$ time.

\begin{algorithm}[bht]
\caption{{\tt full-child-preferences}}\label{alg-ecp}
\begin{algorithmic}[1]
\State \textbf{Input:} An instance of SMTI with children $C$ and families $F$
\State \textbf{Output:} A set $\mathcal{R}$ containing pairs $(i,j')$ that cannot be part of any stable matching
\ForEach{child $i \in C$}
\State ${\cal F}\leftarrow {\emptyset}$
\State ${\cal C}\leftarrow {\emptyset}$
\ForEach{$j \in F(i)$} \Comment{for each family in descending order of preference}
\State ${\cal F} \leftarrow {\cal F} \cup \{j\}$
\State ${\cal C} \leftarrow {\cal C} \cup \{i' \in C(j) : \rankp{j}{i'} \leq \rankp{j}{i}\}$
\If{$|{\cal F}| \geq |{\cal C}|$}
\ForEach{$j' \in {F(i)}$ with $\rank{i}{j'} > \max_{j \in {\cal F}}\{\rank{i}{j}\}$}
\State $\mathcal{R} \leftarrow \mathcal{R} \cup \{(i,j')\}$
\EndForEach
\State \textbf{break}
\EndIf
\EndForEach
\EndForEach
\State \textbf{return} ${\cal R}$
\end{algorithmic}
\end{algorithm}

\rev{We note that: (i) this preprocessing is more powerful when the number of ranks (i.e., groups of tied elements) is high and when there are only a few agents in each rank, and (ii) rather than adding families in descending order of preference, more sophisticated heuristics could find a larger number of reductions at the cost of a higher time complexity.
However, it is worth mentioning that our greedy approach works particularly well when there is a strong correlation between the scores obtained by a given agent among the other agents, e.g., if a family is ranked first for a given child, it also tends to be ranked highly by other children, which is the case in our application.
We show in Section \ref{Comp} that the greedy approaches given by Algorithms \ref{alg-sfr} and~\ref{alg-ecp}} can significantly reduce running times for our SMTI-GRP instances.
We also remark that we did not try to extend Algorithms 1 and 2 to HRT instances with ties on both sides, as our practical application involving SFAS instances allows ties on one side only, and in such a setting we may apply Algorithms ``Hospitals-offer'' and ``Residents-apply'' from \cite{IM09}.

We conclude this section with an example of the application of Algorithms \ref{alg-sfr} and~\ref{alg-ecp}.
\begin{example}
Let us consider an SMTI instance with 5 families and 4 children with the following preference lists:
\begin{alignat*}{8}
c_1: & \quad (f_1~f_2~f_3) & \quad f_4& \quad \quad \quad  & f_1: & \quad (c_1~c_3) && \quad c_4 \\
c_2: & \quad (f_2~f_3~f_4) & \quad f_5 &\quad \quad \quad & f_2: & \quad (c_1~c_2) && \quad c_4 \\
c_3: & \quad (f_1~f_3~f_4) & &\quad \quad \quad & f_3: & \quad (c_2~c_3) && \quad c_1 \\
c_4: & \quad (f_1~f_2~f_4) & &\quad \quad \quad & f_4: & \quad (c_1~c_2) && \quad (c_3~c_4) \\
&&&& f_5: & \quad c_2.
\end{alignat*}

\noindent In this example, child 1 prefers to be matched equally with family 1, 2, and 3.
If his first choice is not granted, then child 1 prefers to be matched with family 4.

We start by running ``first-rank-child", but we see that no two children share the same common set of families as their first preference, so no acceptable pair is removed.
We then run ``first-rank-family'', which highlights that both $f_2$ and $f_4$ have the same pair of children as their equally-first choice ($c_1$ and $c_2$).
This tells us that children $c_1$ and $c_2$ will never be matched with a family that they prefer less than both $f_2$ and $f_4$.
Therefore, there is no need for $c_2$ to ever consider $f_5$.
This leaves the following preferences.
\begin{alignat*}{8}
c_1: & \quad (f_1~f_2~f_3) & \quad f_4& \quad \quad \quad  & f_1: & \quad (c_1~c_3) && \quad c_4 \\
c_2: & \quad (f_2~f_3~f_4) &&\quad \quad \quad & f_2: & \quad (c_1~c_2) && \quad c_4 \\
c_3: & \quad (f_1~f_3~f_4) & &\quad \quad \quad & f_3: & \quad (c_2~c_3) && \quad c_1 \\
c_4: & \quad (f_1~f_2~f_4) & &\quad \quad \quad & f_4: & \quad (c_1~c_2) && \quad (c_3~c_4).
\end{alignat*}

As the instance was reduced, we could now re-run ``first-rank-child'' to see if any further reductions are to be found.
However, no more reductions will be found, and so we move on to ``full-child-preferences'' and ``full-family-preferences''.
We demonstrate the former on child $c_1$ to obtain the following sequence of sets $\mathcal{F}$ and $\mathcal{C}$:
\begin{alignat*}{8}
&{\cal F} = \{f_1\}  &&\quad {\cal C} = \{c_1, c_3\} \\
&{\cal F} = \{f_1, f_2\}  && \quad  {\cal C} = \{c_1, c_2, c_3\} \\
&{\cal F} = \{f_1, f_2, f_3\} &&  \quad  {\cal C} = \{c_1, c_2, c_3\}.
\end{alignat*}
As $|{\cal F}| \geq |{\cal C}|$, we know that $c_1$ cannot be matched with a family that $c_1$ would rank as worse than the worst family in~${\cal F}$.
This means that $c_1$ will never consider $f_4$, so the acceptable pair~$(c_1,f_4)$ can be removed, leaving the following reduced instance.
\begin{alignat*}{8}
c_1: & \quad (f_1~f_2~f_3) && \quad \quad \quad  & f_1: & \quad (c_1~c_3) && \quad c_4 \\
c_2: & \quad (f_2~f_3~f_4) &&\quad \quad \quad & f_2: & \quad (c_1~c_2) && \quad c_4 \\
c_3: & \quad (f_1~f_3~f_4) & &\quad \quad \quad & f_3: & \quad (c_2~c_3) && \quad c_1 \\
c_4: & \quad (f_1~f_2~f_4) & &\quad \quad \quad & f_4: & \quad c_2 && \quad (c_3~c_4).
\end{alignat*}

Since we did reduce the instance, it is possible that re-running one of the other algorithms might reduce the instance even further, but in this particular instance no more reductions can be found.
\end{example}

\section{Reducing the number of non-zero elements}\label{Dum}

Even if the reduction procedures previously described remove a significant number of acceptable pairs, the models involved in real-world instances remain too large to be solved by state-of-the-art ILP solvers. \rev{There are $O(n_1 n_2)$ constraints and variables and up to $O(n_1 n_2 (n_1 + n_2))$ non-zero elements, depending on the length of the agents' preference lists.}
In this section, we propose an alternative formulation for MAX-SMTI that uses dummy variables to keep track of the children's and families' assignments at each rank \rev{so that the overall number of non-zero elements is reduced}.
Let us consider the following additional notation:
\begin{itemize}
\item $\GI{i}$ is the number of distinct ranks (or ties) for child $i$ $(i=1,\dots,n_1)$.
\item $\GPJ{j}$ is the number of distinct ranks for family $j$ $(j=1,\dots,n_2)$.
\item $\FKI{k}{i}$ is the set of families acceptable for child $i$ $(i=1,\dots,n_1)$ with rank $k$ $(k=1,\dots,\GI{i})$.
\item $\CKJ{k}{j}$ is the set of children acceptable for family $j$ $(j=1,\dots,n_2)$ with rank $k$ $(k=1,\dots,\GPJ{j})$.
\end{itemize}

In addition, we introduce the dummy binary decision variables $\YIK{i}{k}$ (respectively, $\YPJK{j}{k}$) that take value $1$ if child $i$ (respectively, family $j$) is matched with a family (respectively, a child) of rank at most $k$, and $0$ otherwise $(i=1,\dots,n_1,\ k=1,\dots,\GI{i})$ (respectively, $j=1,\dots,n_2,\ k=1,\dots,\GPJ{j})$.
Variables $\YIK{i}{k}$ and $\YPJK{j}{k}$ can be seen as a replacement of the summations of $x_{iq}$ and $x_{pj}$ over the sets $\SIJ{i}{j}$ and~$\TJI{j}{i}$.
These variables have certain similarities with the cut-off scores for the college admission problem \cite{ABM16} and the radius formulation for the $p$-median problem~\cite{GLM11}.

The new formulation for MAX-SMTI is:
\begin{alignat}{3}
&\max \quad  && \sum_{i = 1}^{n_1}  \YIK{i}{,\GI{i}} \label{eq:CHF1} \\
&\mbox{~s.t.} && \sum_{j \in \FKI{1}{i}} x_{ij} = \YIK{i}{1}, \quad && i = 1,\dots,n_1, \label{eq:CHF2} \\
& && \sum_{j \in \FKI{k}{i}} x_{ij} + \YIK{i}{, k-1} = \YIK{i}{k}, \quad && i = 1,\dots,n_1,\ k=2,\dots,\GI{i}, \label{eq:CHF3} \\
& && \sum_{i \in \CKJ{1}{j}} x_{ij} = \YPJK{j}{1}, \quad && j = 1,\dots,n_2, \label{eq:CHF4} \\
& && \sum_{i \in \CKJ{k}{j}} x_{ij} + \YPJK{j}{,k-1} = \YPJK{j}{k}, \quad && j = 1,\dots,n_2,\ k=2,\dots,\GPJ{j}, \label{eq:CHF5} \\
& && 1 - \YIK{i}{k} \leq \YPJK{j}{,\rankp{j}{i}}, \quad && i = 1,\dots,n_1,\ k = 1,\dots,\GI{i},\ j \in \FKI{k}{i}, \label{eq:CHF6}\\
& && x_{i,j} \in \{0,1\}, \quad&& i = 1,\dots,n_1,\ j \in F(i), \label{eq:CHF7}\\
& && \YIK{i}{k} \in \{0,1\}, \quad&& i = 1,\dots,n_1,\ k = 1,\dots,\GI{i}, \label{eq:CHF8} \\
& && \YPJK{j}{k} \in \{0,1\}, \quad&& j = 1,\dots,n_2,\ k = 1,\dots,\GPJ{j} \label{eq:CHF9}.
\end{alignat}

The objective function \eqref{eq:CHF1} now uses the last $\YIK{i}{k}$ variable for each child (i.e., the one associated with its last rank) as an indicator of whether the child is assigned to a family.
First, we note that even if \eqref{eq:CHF1} uses fewer non-zero elements than \eqref{eq:C1}, both objective functions are equivalent.
Second, the version of~\eqref{eq:C1} that considers the weight of each pair should be used to solve MAX-WT-SMTI-GRP as \eqref{eq:CHF1} cannot be adapted for the problem.
Constraints \eqref{eq:CHF2}-\eqref{eq:CHF5} maintain the coherence of variables $\YIK{i}{k}$ and $\YPJK{j}{k}$.
Constraints~\eqref{eq:CHF6} ensure the stability of the matching by using the new variables: if child $i$ was not matched with a family of rank $k$ or better (i.e., $1 - \YIK{i}{k} = 1)$, that means that all families that child $i$ ranks at level $k$ were already matched with a child of better or equal rank (i.e., $\YPJK{j}{,\rank{j}{i}} \geq 1 \ \forall j \in \FKI{k}{i}$). Finally, by imposing binary values, constraints \eqref{eq:CHF8}-\eqref{eq:CHF9} prevent any child or family from being matched more than once.
Note that the model would also be valid if variables \YIK{i}{k} and \YPJK{j}{k} were defined as continuous.
However, preliminary experiments showed that it was not beneficial to do so.

Model \eqref{eq:CHF1}-\eqref{eq:CHF9} requires $O(\sum_{i=1}^{n_1} \GI{i} + \sum_{j=1}^{n_2} \GPJ{j})$ additional variables.
It still uses $O(n_1 n_2)$ stability constraints, but they now involve only two variables, which reduces the overall size of the model.

By adopting similar notation for MAX-HRT, where $\GIT{i}$ is the number of ranks (or ties) for doctor $i$ $(i=1,\dots,n_1)$, $\GPJT{j}$ is the number of ranks (ties) for hospital $j$ $(j=1,\dots,n_2)$, $\YIKT{i}{k}$ is a binary decision variable that takes the value $1$ if and only if doctor $i$ is assigned to a hospital of rank at most $k$, and $\YPJKT{j}{k}$ is an integer decision variables indicating how many doctors of rank at most $k$ are assigned to hospital $j$, MAX-HRT becomes:
\begin{alignat}{3}
&\max \quad  && \sum_{i = 1}^{n_1}  \YIKT{i}{,\GIT{i}} \label{eq:modF1} \\
&\mbox{~s.t.} && \YPJKT{j}{\GPJT{j}} \leq c_j, \quad&& j = 1,\dots,n_2, \label{eq:modF6}\\
& && c_j (1 - \YIKT{i}{k}) \leq \YPJKT{j}{,\rankpT{j}{i}}, \quad && i = 1,\dots,n_1,\ k = 1,\dots,\GIT{i},\ j \in \HKI{k}{i}, \label{eq:modF7}\\
& && \YPJKT{j}{k} \in\mathbb{Z}^+, \quad&& j = 1,\dots,n_2,\ k = 1,\dots,\GPJT{j} \label{eq:modF10}, \\
& && (\ref{eq:CHF2}), (\ref{eq:CHF3}), (\ref{eq:CHF4}), (\ref{eq:CHF5}), (\ref{eq:CHF7}), (\ref{eq:CHF8}), \nonumber
\end{alignat}
where (\ref{eq:CHF2})-(\ref{eq:CHF5}) and (\ref{eq:CHF7})-(\ref{eq:CHF8}) are appropriately modified to follow HRT notation.

\section{Alternative stability constraints}\label{Sta}

While dummy variables reduce the number of non-zero elements involved in the stability constraints, we introduce in this section some additional techniques that influence the number of stability constraints and the quality of the continuous relaxations of the models.
\rev{It is well-known that the performance of an integer model depends not only on its size, but also on its linear relaxation.
It was shown in the literature that for several problems (see, e.g., the Bin Packing Problem \cite{DIM16} or the Resource-Constrained Project Scheduling Problem \cite{KALM11}), it may be beneficial to use larger models if they have a better continuous relaxation (i.e., closer to the optimal solution).}

\subsection{Reduced stability constraints for MAX-SMTI}\label{StaS}

\subsubsection{Constraint merging} Model \eqref{eq:CHF1}-\eqref{eq:CHF9} can be further reduced by merging, for a given child, all stability constraints with the same rank. Constraints \eqref{eq:CHF6} now become
\begin{alignat}{5}
& && |\FKI{k}{i}| (1 - \YIK{i}{k}) \leq \sum_{j \in \FKI{k}{i}} \YPJK{j}{,\rankp{j}{i}}, \quad && i = 1,\dots,n_1,\ k = 1,\dots,\GI{i}. \label{eq:CHFP6}
\end{alignat}
This transformation reduces the size of the model, as it uses only $O(\sum_{i=1}^{n_1} \GI{i})$ stability constraints.
However, as will be shown in the computational experiments section, it also leads to a deterioration of the continuous relaxation bound.
We note that the reduction in terms of size with respect to model \eqref{eq:CHF1}-\eqref{eq:CHF9} is more significant when the number of ranks (i.e.,~tie groups) is low.

\subsubsection{Double stability constraints} To compensate for the degradation of the continuous relaxation caused by the constraint merging, it is possible to use the additional stability constraints
\begin{alignat}{5}
& && |\CKJ{k}{j}| (1 -  \YPJK{j}{k}) \leq \sum_{i \in \CKJ{k}{j}} \YIK{i}{,\rank{i}{j}}, \quad && j = 1,\dots,n_2,\ k = 1,\dots,\GPJ{j}. \label{eq:CHFPP6}
\end{alignat}
These constraints can be seen as the counterparts of \eqref{eq:CHFP6} when the merging is performed on the families instead of the children.
These additional constraints improve the quality of the continuous relaxation with respect to the model that uses only \eqref{eq:CHFP6}.
Overall, we observe a tradeoff between the number of stability constraints used in the model and the quality of the bound obtained by the continuous relaxation.

\subsection{New stability constraints for MAX-HRT}\label{StaH}

For MAX-HRT, merging constraint \eqref{eq:modF7} is not useful if there are no ties on the doctors' side (i.e., if $|\HKI{k}{i}| = 1$, $i = 1,\dots,n_1,\ k = 1,\dots,\GIT{i})$.
As our practical case allows ties on the hospitals' side only, it is not an improvement we explored.
In this section, we propose instead an enriched formulation for MAX-HRT that allows us to define a second set of stability constraints.
We introduce new binary decision variables $z_{jk}$ that take value $1$ if hospital $j$ has filled entirely its capacity with doctors of rank at most $k - 1$, and $0$ otherwise $(j=1,\dots,n_2,\ k=1,\dots,\GPJT{j}+1)$.
An additional set of stability constraints for MAX-HRT is:
\begin{alignat}{5}
& & x_{ij} &\leq 1 - z_{j,\rankpT{j}{i}}, \quad&& i = 1,\dots,n_1,\ j \in H(i), \label{eq:S1}\\
& & z_{jk} &\geq z_{j,k-1}, \quad&& j = 1,\dots,n_2,\ k=2,\dots,\GPJT{j}+1, \label{eq:S2}\\
& & 1 - z_{jk} &\leq \YIKT{i}{,\rankT{i}{j}}, \quad&& j = 1,\dots,n_2,\ k=2,\dots,\GPJT{j}+1,\ i \in \DKJ{k-1}{j}, \label{eq:S3}\\
& & c_j z_{j,\GPJT{j}+1} &\leq \YPJKT{j}{,\GPJT{j}}, \quad && j = 1,\dots,n_2, \label{eq:S4} \\
& & z_{jk} & \in \{0,1\}, \quad && j = 1,\dots,n_2,\ k=1,\dots,\GPJT{j}+1 \label{eq:S5}.
\end{alignat}

Constraints \eqref{eq:S1} ensure that a doctor can only be assigned to a hospital that is not already filled by doctors that the hospital strictly prefers.
Constraints \eqref{eq:S2} ensure that, if a hospital is full for doctors of rank at most $k-1$, then it is also full for doctors of rank at most $k$.
Constraints \eqref{eq:S3} ensure that the matching is stable by ruling out the existence of any blocking pair.
More specifically, \eqref{eq:S3} ensure that, if a hospital $j$ has space for doctors of rank $k$ (i.e., $z_{jk} = 0$), then all doctors $i$ of the hospital with rank $k-1$ were already accepted in $j$ or in a hospital they consider equal or better than $j$ (i.e., $\YIKT{i}{,\rankT{i}{j}} = 1$).
Finally, constraints \eqref{eq:S4} ensure that if the hospital is full ($z_{j,\GPJT{j}+1} = 1$), then it has $c_j$ doctors assigned to it ($\YPJKT{j}{,\GPJT{j}} \geq c_j$).

As ties occur on the hospital side, constraint merging can be applied to \eqref{eq:S3} to obtain:
\begin{alignat}{5}
& & |\DKJ{k-1}{j}| (1 - z_{jk}) &\leq \sum_{i \in \DKJ{k-1}{j}} \YIKT{i}{,\rankT{i}{j}}, \quad&& j = 1,\dots,n_2,\ k=2,\dots,\GPJT{j}+1. \label{eq:SP3}
\end{alignat}

Note that both sets of stability constraints \eqref{eq:modF7} and \eqref{eq:S1}-\eqref{eq:S5} can be used at the same time. Moreover, stability constraints
\begin{alignat}{5}
& & c_j z_{jk} \leq \YPJKT{j}{,k-1}, \quad&& j = 1,\dots,n_2,\ k=2,\dots,\GPJT{j}+1, \label{eq:mix1}
\end{alignat}
stating that if a given hospital $j$ has no room for doctors at rank $k$ (i.e., $z_{jk}=1$), then it has already selected $c_j$ doctors of rank at most $k-1$ (i.e., $\YPJKT{j}{,k-1}\geq c_j$), could be used to replace~\eqref{eq:modF7}.
Indeed, let us consider a doctor $i$, a hospital~$j$, and their respective ranks $\rankT{i}{j} = m$ and $\rankpT{j}{i} = p$.
By \eqref{eq:S3}, we know that
\begin{alignat}{5}
1 - z_{j,p+1} \leq \YIKT{i}{m},
\end{alignat}
which can be rewritten as
\begin{alignat}{5}
c_j (1 - \YIKT{i}{m}) \leq c_j z_{j,p+1} \label{eq:mix3}.
\end{alignat}
This can be completed by \eqref{eq:mix1} to obtain
\begin{alignat}{5}
c_j (1 - \YIKT{i}{m}) \leq c_j z_{j,p+1} \leq \YPJKT{j}{p},
\end{alignat}
which leads to constraints \eqref{eq:modF7} being redundant.
Notice that while we count $O(n_1 n_2)$ stability constraints \eqref{eq:modF7}, only $O(\sum_{j=1}^{n_2} \GPJT{j})$ are required for \eqref{eq:mix1}.

\section{Computational experiments}\label{Comp}

We report in this section the outcome of extensive computational experiments aimed at~testing the effectiveness of the proposed improvements for MAX-WT-SMTI-GRP, MAX-SMTI, and MAX-HRT.
All algorithms were coded in {\tt C++}, and Gurobi 7.5.2 was used to solve the ILP models.
The implemented software is downloadable from the online repository \url{https://dx.doi.org/10.5281/zenodo.2538150}.
The experiments were run on an Intel Xeon E5-2687W v3, 3.10GHz with 512GB of memory, running under Linux 4.13.0.
Each instance was run using a single core and had a total time limit (comprising model creation time and solution time) of 3600~seconds per problem instance.
The instances that were randomly generated are downloadable from the online repository \url{https://dx.doi.org/10.5525/gla.researchdata.664}.

For each problem, a first set of experiments determined what combination of the improvements proposed in Sections \ref{Dum} and \ref{Sta} is the most effective.
At a second stage, we reran a subset of these combinations to evaluate the impact of other features (such as preprocessing, branching priorities, and warm start).
\rev{Experimental evaluations of the algorithms for MAX-WT-SMTI-GRP, MAX-SMTI, and MAX-HRT are presented in Sections \ref{s1}, \ref{s2}, and \ref{s3} respectively.
A~summary of the experiments reported in this section is presented in Table \ref{table-0}.
For each experiment, the table identifies the problem solved, the subsequent tables containing the results, the dataset used (number of instances and source), and the experimental objectives.
We show that dummy variables are particularly useful for MAX-WT-SMTI-GRP, constraint merging is beneficial for MAX-SMTI, and the new set of stability constraints substantially improves the performance of MAX-HRT. A more detailed discussion of these results can be found in Section~\ref{s4}}

\begin{table}[hbt]
	\scriptsize
	\begin{center}
    	\caption{\rev{Summary of the experiments}} \label{table-0}
	\setlength{\tabcolsep}{0.2cm}
	\begin{tabular}{lrrllllrrrrrrrrrrrr}
     	\toprule
        \multicolumn{1}{c}{\multirow{2}{*}{Problem}} &  \multicolumn{1}{c}{\multirow{2}{*}{Table}} & \multicolumn{2}{c}{\multirow{1}{*}{Dataset}} & \multicolumn{1}{c}{\multirow{2}{*}{Purpose}} \\
       \cmidrule(r){3-4}
         & &  \multicolumn{1}{c}{\#inst.} & \multicolumn{1}{c}{Source} & \\
        \cmidrule(r){1-1} \cmidrule(r){2-2} \cmidrule(r){3-4} \cmidrule(r){5-5}
\multirow{3}{*}{MAX-WT-SMTI-GRP} & 3 & 22 & Coram & impact of dummy var. and alternative stability cons. \\
& 4 & 22 & Coram & impact of preprocessing, warm start, and priorities \\
& 5 & 220 & randomly generated & models' limits \\
\cmidrule(r){1-1} \cmidrule(r){2-2} \cmidrule(r){3-4} \cmidrule(r){5-5}
MAX-SMTI & 6--8 & 270 & randomly generated & impact of dummy var. and alternative stability cons. \\
\cmidrule(r){1-1} \cmidrule(r){2-2} \cmidrule(r){3-4} \cmidrule(r){5-5}
\multirow{4}{*}{MAX-HRT non master} & 9 & 3 & SFAS & impact of dummy var. and alternative stability cons. \\
  & 10 & 700 & randomly generated & impact of dummy var. and alternative stability cons. \\
  & 11 & 150 & randomly generated & models' limits \\
  & 12 & 60 & randomly generated & impact of preprocessing, warm start, and priorities\\
\cmidrule(r){1-1} \cmidrule(r){2-2} \cmidrule(r){3-4} \cmidrule(r){5-5}
MAX-HRT master & 13 & 450 & randomly generated & difference between the initial and the best model  \\
 \bottomrule
\end{tabular}
\end{center}
\end{table}
\subsection{MAX-WT-SMTI-GRP} \label{s1}

\subsubsection{Real-world instances}

We were provided with a sample set of data representing 550 children and 894 families, which included a weight (determined by Coram) for every child-family pair.
Since most of the weights vary between 80 and 100, for each integer threshold in $\{0, 80, 81, \ldots,100\}$, we created an instance of SMTI-GRP as described in Section \ref{Prob2}, resulting in 22 instances.
When the threshold is set to 0, all child-family pairs are acceptable, and so technically we have complete rather than incomplete preference lists.
Our models are still applicable to such instances.

The sample dataset contained a significant number of ties.
One way to measure the density of ties in an SMTI or HRT instance is now described.
Given an instance of SMI (i.e.,~with no ties), an instance of SMTI can be created as follows.
For each set of agents pick a {\em tie density} $t_d$ with $0 \leq t_d \leq 1$ (i.e.,~so the children and families may have distinct tie densities).
Then let any two consecutive elements in any preference list from this set of agents be tied with probability $t_d$ ($0\leq t_d \leq 1$)~\cite{KM14}.
We reverse this procedure here, taking a real-world instance and calculating what proportion of ties exist on each side of it using the following process.
First, count the number of distinct tie groups $g$ and the number of actual elements $e$ in the preference lists on one side, and let $n$ be the number of agents on that side that have at least one agent in their preference list.
The tie density of that side is then given by $t_d = 1-\frac{g-n}{e-n}$.
Subtracting $n$ from the numerator ensures that if the agents consider all possibilities equally, we obtain a tie density of~1, and subtracting $n$ from the denominator ensures that an absence of ties equates to a density of 0.
Note that $e > n$ is assumed, that is, at least one agent has more than two agents in its preference list.
If $e = n$, the instance is trivial to solve.

Through this formula, the tie density is 0.9716 for the families and 0.9705 for the children.
Among the problem instances generated in~\cite{KM14}, it was found that those with a tie density $t_d \approx 0.85$ tended to be the most challenging to solve.

For a given child or family, the variance of weights is relatively low: 3.8 on average for the children and 3.7 for the families.
This suggests that some children are considered ``good" or ``bad" matches for many families, and vice-versa.

Each of the almost half a million child-family pairs is given a weight, but there are only 54 distinct weights in total.
A bar chart displaying these weights is shown in Figure \ref{Fig-1} where we observe two distinct behaviours: some values appear many times (we will call them ``common'') and some values appear just a few times in comparison (we will call them ``uncommon").

\begin{figure}[hbt]
\centering
\caption{Bar chart of the child-family weights for the complete instance}
\includegraphics[scale=0.425]{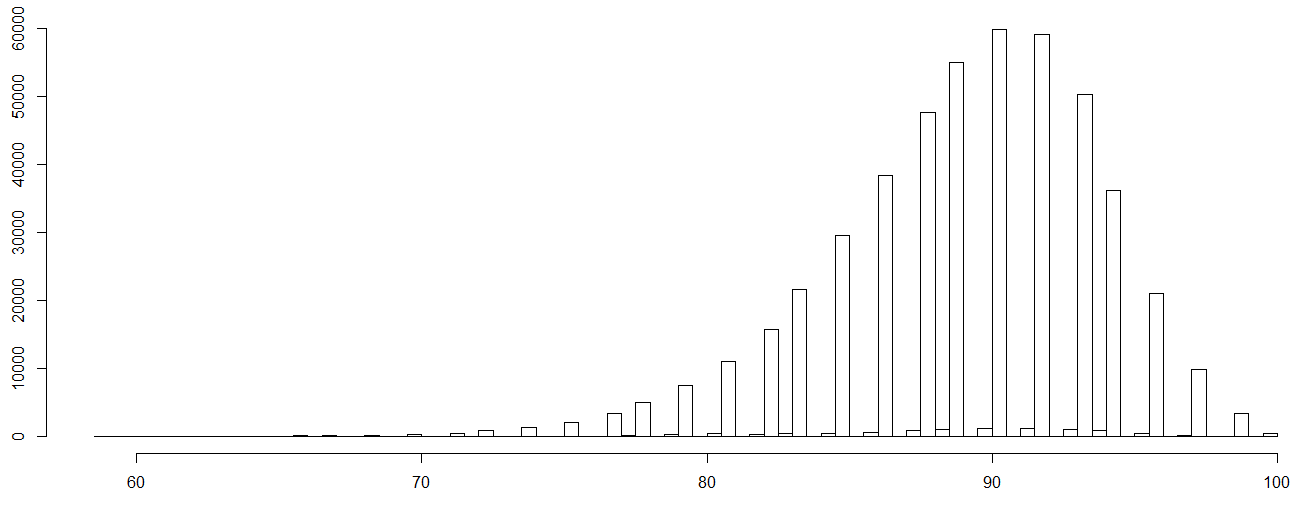}
\label{Fig-1}
\end{figure}
We examine our improvements to the models on the 22 SMTI-GRP real-world instances generated from the sample data set with the different thresholds.
In all methods we considered, the preprocessing described in Section \ref{Pre} was applied.
Table \ref{table-1} compares the six possible combinations of the proposed improvements.
The ``Method'' columns detail the combination of options, with some attributes describing the specific implementation: ``index'' identifies the method while ``dummy variables'', ``stability constraint merging'', and ``double stability constraints'' indicate the inclusion or otherwise of the corresponding feature in the model.
The three following columns give some indicators of the performance of each method: the number of optimal solutions found, average CPU time over all runs (including the ones terminated by the time limit), and the continuous relaxation value.
The three last columns report some details about the model size: average number of variables, constraints, and non-zeros elements.

\begin{table}[hbt]
	\scriptsize
	\begin{center}
    	\caption{Comparison of the proposed methods for preprocessed SMTI-GRP real-world instances} \label{table-1}
	\setlength{\tabcolsep}{0.2cm}
	\begin{tabular}{ccccrrrrrrrrrrrrrrr}
     	\toprule
        \multicolumn{4}{c}{\multirow{1}{*}{Method}} &  \multicolumn{3}{c}{{Values}} & \multicolumn{3}{c}{{Model size}} \\
        \cmidrule(r){1-4} \cmidrule(r){5-7} \cmidrule(r){8-10}
        index & \makecell{dummy\\variables} & \makecell{stab. cons. \\ merging} & \makecell{double\\stab. cons.} & \#opt &  \makecell{time}   &  \makecell{continuous\\relaxation}  & \makecell{number of\\variables}  & \makecell{number of\\constraints} & \makecell{number of\\non-zeros} \\
        \cmidrule(r){1-4} \cmidrule(r){5-7} \cmidrule(r){8-10}
M1 &	&	&	&	21&	434.0&	 \num{42966.0}&	 \num{94764}&	 \num{96208}&	 \num{39152977}\\
M2 &	x &	&	&	22&	86.8&		 \num{42966.0}&	 \num{101220}&	 \num{101220}&	 \num{390790}\\
M3 &	&	x &	&	20&	416.3&		 \num{43030.4}&	 \num{94764}&	 \num{3457}&	 \num{13619285}\\
M4 &	x &	x &	&	22&	73.7&		 \num{43030.4}&	 \num{101220}&	 \num{8468}&	 \num{298038}\\
M5 &	&	x &	x &	20&	722.0&		 \num{43010.0}&	 \num{94764}&	 \num{7899}&	 \num{39853723}\\
M6 &	x &	x &	x &	22&	66.2&		 \num{43010.0}&	 \num{101220}&	 \num{12911}&	 \num{397245}\\
 \bottomrule
\end{tabular}
\end{center}
\end{table}
The table shows a number of interesting facts:
\begin{itemize}
\item The real-world instances did not challenge our algorithms as even the basic model M1, (introduced in \cite{KM14}) solves 21 out of the 22 instances.
\item Only the algorithms using dummy variables solve all the real-world instances.
    Algorithms using dummy variables are one order of magnitude faster than those that do not (e.g., 434~seconds on average for M1 versus 86.8 for M2).
    This is probably due to the decrease by two orders of magnitude in the average number of non-zero elements.
\item Merging stability constraints seems to have a positive effect in terms of CPU time (434~seconds on average for M1 versus $416.3$ for M3, and $86.8$ seconds for M2 versus 73.7 for M4), but it remains marginal compared to the use of dummy variables.
    While using this feature worsens the continuous relaxation value (e.g., from $\num{42966}$ for M1 to $\num{43030.4}$ for M3), it also decreases the model size (e.g., from $\num{96208}$ constraints for M1 to $\num{3457}$ for M3).
\item Provided that constraint merging is used, it is unclear if using double stability constraints is beneficial: while it increases the average time from $416.3$ seconds for M3 to $722$ for M5, it also decreases the average time from $73.7$ seconds for M4 to $66.2$ for M6.
    In contrast to constraint merging, the use of double stability constraints improves the continuous relaxation value at the expense of increasing the model size.
\end{itemize}
\rev{Supplementary experiments showed that, by dropping the stability constraints, the optimal objective value of our MAX-WT-SMTI-GRP real-world instances would increase by 3\% on average (with a minimum of 0\% and a maximum of 18\%).
Without stability constraints, the problem can be solved in polynomial time within seconds \cite{GT89}.}

Table \ref{table-2} studies the impact of the preprocessing discussed in Section \ref{Pre}, the use of the Gale-Shapley algorithm \cite{GS62} to give a warm-start to the solver, and the use of priorities on the dummy variables during the branch-and-bound process.
We tested these features on the basic model~(M1), and one of the fastest algorithms (M4). We selected M4 because it obtained very good results while keeping a relatively low number of variables, constraints, and non-zero elements.
In our experiments involving a warm start, we called the Gale-Shapley algorithm 100~times after breaking ties randomly, and fed the solver with the best solution found.

\begin{table}[H]
	\scriptsize
	\begin{center}
    	\caption{Comparison of M1 and M4 with additional features for SMTI-GRP real-world instances} \label{table-2}
	\setlength{\tabcolsep}{0.2cm}
	\begin{tabular}{ccccrrrrrrrrrrrrrrr}
     	\toprule
        \multicolumn{4}{c}{\multirow{1}{*}{Method}} &  \multicolumn{2}{c}{{Values}} & \multicolumn{3}{c}{{Model size}} \\
        \cmidrule(r){1-4} \cmidrule(r){5-6} \cmidrule(r){7-9}
        index & \makecell{preprocessing} & \makecell{warm\\start} & \makecell{dummy variables \\ priorities} & \#opt &  \makecell{time} & \makecell{number of\\variables}  & \makecell{number of\\constraints} & \makecell{number of\\non-zeros} \\
        \cmidrule(r){1-4} \cmidrule(r){5-6} \cmidrule(r){7-9}
M1 &	&	&	&	6&	1325.0&			 \num{237666}&	 \num{239110}&	 \num{162252108}\\
M1 &	x &	&	&	21&	434.0&			 \num{94764}&	 \num{96208}&	 \num{39152977}\\
M1 &	x &	x &	&	21&	288.9&			 \num{94764}&	 \num{96208}&	 \num{39152977}\\
M4 &	 &	&	&	22&	155.4&			 \num{248442}&	 \num{15689}&	 \num{738240}\\
M4 &	x &	 &	 &	22&	73.7&			 \num{101220}&	 \num{8468}&	 \num{298038}\\
M4 &	x &	x &	 &	22&	59.2&			 \num{101220}&	 \num{8468}&	 \num{298038}\\
M4 &	x &	&	x &	22&	76.7&			 \num{101220}&	 \num{8468}&	 \num{298038}\\
 \bottomrule
\end{tabular}
\end{center}
\end{table}
The results in Table \ref{table-2} show that:
\begin{itemize}
\item Without preprocessing, the basic model M1 can only solve 6 out of 22 instances.
\item Preprocessing leads to the removal of more than half of the variables, about half of the constraints, and more than half of the non-zero elements, both for models M1 and M4.
\item Applying Gale-Shapley algorithm is useful as it reduces the average CPU time of M1 from 434 seconds to 288.9 and the average CPU time of M4 from 73.7 seconds to 59.2.
\item Giving a high priority to the dummy variables $\YIK{i}{k}$ and $\YPJK{j}{k}$ during the branching process of~M4 does not help the solver.
    As this feature did not seem promising, we decided not to try further combinations of M4.
    We also mention that since M1 does not have dummy variables, the model could not be used to test this feature.
\end{itemize}

\subsubsection{Augmented instances}

As all real-world instances could be solved within an hour by the best method, we generated bigger instances to test the limits of our models.
In order to keep the properties of the original instance, we used it as a basis to create the new benchmark.

We started by duplicating each child and each family $\kappa$ times.
Then, we perturbed the weight obtained by each pair in the following way: (1) identify the group of the weight (common or uncommon) and the position $i$ of the weight in that group, and (2), change the weight to the one in position $i+x$ with probability $p(x)$ for $(x,p(x)) \in \{(-2,0.1), (-1,0.2), (0, 0.4), (1, 0.2), (2, 0.1)\}$.

For each $\kappa \in \{1, 2\}$, we generated 10 instances, to which we applied the 22 thresholds, resulting in 440 new instances in total.
Table \ref{table-D} compares the six possible combinations of the proposed improvements, with preprocessing, on the 440 augmented SMTI-GRP instances.

\begin{table}[hbt]
	\scriptsize
	\begin{center}
    	\caption{Comparison of the methods for preprocessed SMTI-GRP augmented instances} \label{table-D}
	\setlength{\tabcolsep}{0.15cm}
	\begin{tabular}{ccccrrrrrrrrrrrrrrr}
     	\toprule
         \multicolumn{1}{c}{\multirow{2}{*}{Index}} &
        \multicolumn{6}{c}{{$\kappa$ = 1}} & \multicolumn{6}{c}{{$\kappa$ = 2}}  \\
    \cmidrule(r){2-7} \cmidrule(r){8-13}
&  \#opt &  \makecell{time}   &  \makecell{cont.\\relax.}  & \makecell{nb. \\ var.}  & \makecell{nb.\\cons.} & \makecell{nb.\\nzs.} &
\#opt &  \makecell{time}   &  \makecell{cont.\\relax.}  & \makecell{nb. \\ var.}  & \makecell{nb.\\cons.} & \makecell{nb.\\nzs.} \\
 \cmidrule(r){1-1}  \cmidrule(r){2-7} \cmidrule(r){8-13}
M1 &				220&	340.6&	 \num{45460.0}&	 \num{63974}&	 \num{65418}&	 \num{21214917}&	74&	2701.5&	 \num{92341.5}&	 \num{358549}&	 \num{361437}&	 \num{278330845}\\
M2 &				220&	114.8&	 \num{45460.0}&	 \num{72542}&	 \num{72542}&	 \num{271809}&	121&	2185.1&	 \num{92341.5}&	 \num{379695}&	 \num{379695}&	 \num{1473987}\\
M3 &				220&	201.5&	 \num{45494.8}&	 \num{63974}&	 \num{3770}&	 \num{6742990}&	52&	2962.5&	 \num{92476.6}&	 \num{358549}&	 \num{8865}&	 \num{94883495}\\
M4 &				220&	84.0&	 \num{45494.8}&	 \num{72542}&	 \num{10894}&	 \num{210162}&	134&	2114.1&	 \num{92476.6}&	 \num{379695}&	 \num{27122}&	 \num{1121414}\\
M5 &				220&	507.9&	 \num{45484.5}&	 \num{63974}&	 \num{10012}&	 \num{21874731}&	43&	3046.8&	 \num{92450.1}&	 \num{358549}&	 \num{24034}&	 \num{282418566}\\
M6 &				220&	55.9&	 \num{45484.5}&	 \num{72542}&	 \num{17136}&	 \num{280377}&	143&	2106.7&	 \num{92450.1}&	 \num{379695}&	 \num{42291}&	 \num{1495132}\\
 \bottomrule
\end{tabular}
\end{center}
\end{table}

We observe that the comments made about the real-world instances are still valid here: using dummy variables reduces significantly the number of non-zero elements and improves the performance of the model.
Constraint merging seems to have a positive effect, especially when used together with dummy variables.
The same behaviour is observed for double stability constraints.
Overall, we notice a clear improvement from the basic model M1 to the more sophisticated model M6, as the former could only solve 74 out of the 220 instances with $\kappa = 2$, while the latter could solve 143 of them.
Despite this remarkable improvement, some instances with $\kappa = 2$ (i.e., with 1100 children and 1788 families) remain unsolved after one hour of computing time.
We also generated instances with $\kappa = 3$, but running these experiments was impractical due to the large memory requirements of the models.

\subsection{MAX-SMTI}  \label{s2}

Even though we initially developed our models for MAX-WT-SMTI-GRP, they can also be used to solve MAX-SMTI.
As, to the best of our knowledge, no \rev{MAX-SMTI datasets} are available in the literature, we used the generator described in \cite{IM09} to create new instances in order to test the effectiveness of our methods on this problem too.
We tried several combinations of number of agents (\{$\num{10000}, \num{25000}, \num{50000}\}$ on each side), tie density ($\{0.75, 0.85, 0.95\}$ on each side), and preference list length ($\{3, 5, 10\}$ on one side, as the generator does not support a limit on the preference list lengths on both sides).
We generated 10 instances for each combination resulting in 270 SMTI instances in total.

Tables \ref{table-A} to \ref{table-C} compare the six possible combinations of the proposed improvements, with preprocessing, on the 270 random SMTI instances.

\begin{table}[b]
	\scriptsize
	\begin{center}
    	\caption{Comparison of the methods for preprocessed SMTI instances with $\num{10000}$ agents} \label{table-A}
	\setlength{\tabcolsep}{0.1cm}
	\begin{tabular}{ccccrrrrrrrrrrrrrrr}
     	\toprule
        \multicolumn{1}{c}{\multirow{3}{*}{Index}} & \multicolumn{6}{c}{preference list length = 3} & \multicolumn{6}{c}{preference list length = 5} & \multicolumn{6}{c}{preference list length = 10} \\
         \cmidrule(r){2-7} \cmidrule(r){8-13} \cmidrule(r){14-19}
        & \multicolumn{2}{c}{0.75} & \multicolumn{2}{c}{0.85}  & \multicolumn{2}{c}{0.95}
        & \multicolumn{2}{c}{0.75} & \multicolumn{2}{c}{0.85}  & \multicolumn{2}{c}{0.95}
        & \multicolumn{2}{c}{0.75} & \multicolumn{2}{c}{0.85}  & \multicolumn{2}{c}{0.95}   \\
    \cmidrule(r){2-3} \cmidrule(r){4-5} \cmidrule(r){6-7}  \cmidrule(r){8-9}  \cmidrule(r){10-11} \cmidrule(r){12-13} \cmidrule(r){14-15} \cmidrule(r){16-17} \cmidrule(r){18-19}
&  \#opt &  \makecell{time}   &   \#opt &  \makecell{time}   & \#opt &  \makecell{time} &
  \#opt &  \makecell{time}   &   \#opt &  \makecell{time}   & \#opt &  \makecell{time} &
  \#opt &  \makecell{time}   &   \#opt &  \makecell{time}   & \#opt &  \makecell{time} \\
        \cmidrule(r){1-1} \cmidrule(r){2-3} \cmidrule(r){4-5} \cmidrule(r){6-7}  \cmidrule(r){8-9}  \cmidrule(r){10-11} \cmidrule(r){12-13} \cmidrule(r){14-15} \cmidrule(r){16-17} \cmidrule(r){18-19}
M1&	10&	10&	10&	8&	10&	8&	10&	155&	10&	94&	10&	39&	0&	3600&	4&	3428&	10&	382\\
M2&	10&	11&	10&	9&	10&	7&	10&	361&	10&	134&	10&	24&	0&	3600&	0&	3600&	10&	288\\
M3&	10&	5&	10&	5&	10&	5&	10&	75&	10&	38&	10&	23&	9&	2731&	9&	1555&	10&	490\\
M4&	10&	10&	10&	9&	10&	7&	10&	150&	10&	43&	10&	18&	1&	3543&	10&	1723&	10&	95\\
M5&	10&	10&	10&	8&	10&	8&	10&	115&	10&	74&	10&	41&	0&	3601&	10&	2073&	10&	525\\
M6&	10&	10&	10&	8&	10&	6&	10&	442&	10&	163&	10&	29&	0&	3600&	0&	3600&	10&	496\\
 \bottomrule
\end{tabular}
\end{center}
\end{table}

\begin{table}[t]
	\scriptsize
	\begin{center}
    	\caption{Comparison of the methods for preprocessed SMTI instances with $\num{25000}$ agents} \label{table-B}
	\setlength{\tabcolsep}{0.1cm}
	\begin{tabular}{ccccrrrrrrrrrrrrrrr}
     	\toprule
        \multicolumn{1}{c}{\multirow{3}{*}{Index}} & \multicolumn{6}{c}{preference list length = 3} & \multicolumn{6}{c}{preference list length = 5} & \multicolumn{6}{c}{preference list length = 10} \\
         \cmidrule(r){2-7} \cmidrule(r){8-13} \cmidrule(r){14-19}
        & \multicolumn{2}{c}{0.75} & \multicolumn{2}{c}{0.85}  & \multicolumn{2}{c}{0.95}
        & \multicolumn{2}{c}{0.75} & \multicolumn{2}{c}{0.85}  & \multicolumn{2}{c}{0.95}
        & \multicolumn{2}{c}{0.75} & \multicolumn{2}{c}{0.85}  & \multicolumn{2}{c}{0.95}   \\
    \cmidrule(r){2-3} \cmidrule(r){4-5} \cmidrule(r){6-7}  \cmidrule(r){8-9}  \cmidrule(r){10-11} \cmidrule(r){12-13} \cmidrule(r){14-15} \cmidrule(r){16-17} \cmidrule(r){18-19}
&  \#opt &  \makecell{time}   &   \#opt &  \makecell{time}   & \#opt &  \makecell{time} &
  \#opt &  \makecell{time}   &   \#opt &  \makecell{time}   & \#opt &  \makecell{time} &
  \#opt &  \makecell{time}   &   \#opt &  \makecell{time}   & \#opt &  \makecell{time} \\
        \cmidrule(r){1-1} \cmidrule(r){2-3} \cmidrule(r){4-5} \cmidrule(r){6-7}  \cmidrule(r){8-9}  \cmidrule(r){10-11} \cmidrule(r){12-13} \cmidrule(r){14-15} \cmidrule(r){16-17} \cmidrule(r){18-19}
M1&	10&	65&	10&	53&	10&	45&	10&	1424&	10&	686&	10&	199&	0&	3600&	0&	3601&	4&	3439\\
M2&	10&	61&	10&	48&	10&	35&	6&	3329&	10&	928&	10&	144&	0&	3600&	0&	3600&	6&	3310\\
M3&	10&	23&	10&	24&	10&	23&	10&	583&	10&	230&	10&	117&	0&	3600&	0&	3601&	3&	3234\\
M4&	10&	56&	10&	46&	10&	32&	10&	1527&	10&	344&	10&	89&	0&	3600&	0&	3600&	10&	739\\
M5&	10&	55&	10&	46&	10&	38&	10&	873&	10&	429&	10&	170&	0&	3601&	0&	3601&	7&	2960\\
M6&	10&	54&	10&	44&	10&	31&	0&	3600&	10&	1449&	10&	156&	0&	3600&	0&	3600&	0&	3600\\
 \bottomrule
\end{tabular}
\end{center}
\end{table}

\begin{table}[t]
	\scriptsize
	\begin{center}
    	\caption{Comparison of the methods for preprocessed SMTI instances with $\num{50000}$ agents} \label{table-C}
	\setlength{\tabcolsep}{0.1cm}
	\begin{tabular}{ccccrrrrrrrrrrrrrrr}
     	\toprule
        \multicolumn{1}{c}{\multirow{3}{*}{Index}} & \multicolumn{6}{c}{preference list length = 3} & \multicolumn{6}{c}{preference list length = 5} & \multicolumn{6}{c}{preference list length = 10} \\
         \cmidrule(r){2-7} \cmidrule(r){8-13} \cmidrule(r){14-19}
        & \multicolumn{2}{c}{0.75} & \multicolumn{2}{c}{0.85}  & \multicolumn{2}{c}{0.95}
        & \multicolumn{2}{c}{0.75} & \multicolumn{2}{c}{0.85}  & \multicolumn{2}{c}{0.95}
        & \multicolumn{2}{c}{0.75} & \multicolumn{2}{c}{0.85}  & \multicolumn{2}{c}{0.95}   \\
    \cmidrule(r){2-3} \cmidrule(r){4-5} \cmidrule(r){6-7}  \cmidrule(r){8-9}  \cmidrule(r){10-11} \cmidrule(r){12-13} \cmidrule(r){14-15} \cmidrule(r){16-17} \cmidrule(r){18-19}
&  \#opt &  \makecell{time}   &   \#opt &  \makecell{time}   & \#opt &  \makecell{time} &
  \#opt &  \makecell{time}   &   \#opt &  \makecell{time}   & \#opt &  \makecell{time} &
  \#opt &  \makecell{time}   &   \#opt &  \makecell{time}   & \#opt &  \makecell{time} \\
        \cmidrule(r){1-1} \cmidrule(r){2-3} \cmidrule(r){4-5} \cmidrule(r){6-7}  \cmidrule(r){8-9}  \cmidrule(r){10-11} \cmidrule(r){12-13} \cmidrule(r){14-15} \cmidrule(r){16-17} \cmidrule(r){18-19}
M1&	10&	230&	10&	202&	10&	146&	0&	3600&	10&	2565&	10&	523&	0&	3601&	0&	3601&	0&	3601\\
M2&	10&	239&	10&	209&	10&	146&	0&	3600&	2&	3336&	10&	407&	0&	3600&	0&	3600&	0&	3601\\
M3&	10&	74&	10&	78&	10&	73&	10&	1966&	10&	828&	10&	291&	0&	3601&	0&	3601&	0&	3602\\
M4&	10&	220&	10&	187&	10&	134&	0&	3600&	10&	1046&	10&	220&	0&	3600&	0&	3600&	10&	2132\\
M5&	10&	212&	10&	185&	10&	142&	10&	2684&	10&	1227&	10&	339&	0&	3597&	0&	3601&	0&	3603\\
M6&	10&	209&	10&	190&	10&	132&	0&	3600&	0&	3600&	10&	401&	0&	3600&	0&	3600&	0&	3600\\
 \bottomrule
\end{tabular}
\end{center}
\end{table}

We observe that:
\begin{itemize}
\item As expected, for a given tie density and preference list length, a larger number of agents results in harder instances.
\item For a given number of agents and a given preference list length, instances with a tie density at 0.75 are harder than those at 0.85, which are themselves harder than those at 0.95.
     This could be explained by the fact that the difference between the continuous relaxation and the optimal solution size is smaller as the tie density increases.
     For example, for M1 with $\num{25000}$ agents and a preference list length of 3, we observed an average absolute difference between the two values of 5.07 when the tie density was 0.75, 2.81  when the tie density was 0.85, and 0.15 when the tie density was 0.95.
\item For a given number of agents and tie density, shorter preference lists make the instances easier.
    This is probably due to the fact that the models have fewer variables, constraints, and non-zero elements.
\item Once again, the effectiveness of the new models is demonstrated: out of the 30 instances with preference list length of 5 and $\num{50000}$ agents, M1 solves 20 instances in total while M3 solves all 30 instances.
\item For each combination, the best results were obtained either by M3 or by M4.
    We distinguish several cases:
    \begin{itemize}
    \item When the preference list length is equal to 3, M3 is always the best.
        Indeed, the preference lists are too short to obtain any benefit from the dummy variables.
        For example:
        \begin{itemize}
        \item For tie density 0.85, preference list length of 3, and $\num{25000}$ agents, M3 (resp.~M4) has on average \num{72227} variables (resp.~\num{133976}), \num{81322} constraints (resp.~\num{93071}), and \num{408556} non-zero elements (resp. \num{322817}, i.e., 21\% less).
        \end{itemize}
    \item When the preference list length is equal to 5, M3 is always the best for tie densities 0.75 and 0.85 while M4 is always the best for tie density 0.95.
        This could be explained by the fact that high tie densities involve fewer tie groups.
        Thus, fewer additional constraints and variables are required by the dummy variables and using them saves more non-zero elements.
        For example:
        \begin{itemize}
        \item For tie density 0.85, preference list length of 5, and $\num{25000}$ agents, M3 (resp. M4) has on average \num{119551} variables (resp. \num{196273}), \num{88448} constraints (resp. \num{115171}), and \num{875334} non-zero elements (resp. \num{500736}, i.e., 43\% less).
        \item For tie density 0.95, preference list length of 5, and $\num{25000}$ agents, M3 (resp. M4) has on average \num{124223} variables (resp. \num{183693}), \num{79797} constraints (resp. \num{89267}), and \num{965056} non-zero elements (resp. \num{471574}, i.e., 51\% less).
        \end{itemize}
    \item When the preference list length is equal to 10, M3 is the best for tie density 0.75, while M4 is the best for tie densities 0.85 and 0.95.
        Again, the use of dummy variables is shown to be beneficial with longer preference lists.
    \end{itemize}
\end{itemize}
We also report that additional computational experiments showed that: \rev{(i) the preprocessing techniques introduced in Section \ref{Pre} have little effect on these instances, and (ii) dropping the stability constraints would increase the size of the matching by at most 1\% and make the problem polynomial-time solvable \cite{HK73}.}

\subsection{MAX-HRT}  \label{s3}
In many instances of SMTI and HRT, it can be assumed that agents establish their ranking based on their own individual preferences.
However, sometimes it is the case that agents' preferences are formulated on the basis of objective criteria.
For example, in a specific version of MAX-HRT, hospitals only consider the grades of the doctors for their preference lists.
In this situation, the so-called {\em master list}, a ranking of all the doctors based on their grades, is made at the beginning.
The preference list of each hospital is then an exact copy of the master list from which the doctors who did not apply to the given hospital were removed.
We tested our algorithms on instances both with and without a master list of doctors.

\subsubsection{Non master list instances}

We had access to instances of the Scottish Foundation Allocation Scheme, which assigned medical graduates to Scottish hospitals, for the years 2006, 2007, and 2008 \cite{IM09}.
These 3 instances, called ``SFAS" in the following, have respectively 759, 781 and 748 doctors, 53, 53 and 52 hospitals, and 801, 789 and 752 available positions.
Doctors chose exactly 6 hospitals, although a small number of exceptions with shorter preference lists were found.
No master list was used by the hospitals to establish their preference lists, even if we observed a tendency for some doctors to be often well (or badly) ranked.
The tie density on the hospitals' side was 0.9468, 0.7861 and 0.8424 respectively.
The tie density on the doctors' side was 0 as the doctors were asked to provide strict preferences.

We also tested our algorithms on the dataset ``SET2'' described in \cite{KM14} and available at \url{http://researchdata.gla.ac.uk/244/}.
It comprises 700 instances with 100, 150, 200, 250, 300, 350, 400 doctors, and 7, 10, 14, 17, 21, 24, 28 hospitals, respectively.
The number of available positions was exactly equal to the number of doctors.
Doctors chose exactly 5 hospitals, and the tie density was equal to 0.85.

Table \ref{table-4} compares twelve combinations of the proposed improvements, with preprocessing, on the 3 SFAS real-world instances.
The meaning of the columns is unchanged with respect to Table \ref{table-1}, except for ``stability constraints", which now indicates the set of stability constraints used by the model, and ``stability constraint merging", which now indicates whether constraints~\eqref{eq:S3} were merged or not.

\begin{table}[bht]
	\scriptsize
	\begin{center}
    	\caption{Comparison of the proposed methods for preprocessed HRT SFAS instances} \label{table-4}
	\setlength{\tabcolsep}{0.2cm}
	\begin{tabular}{ccccrrrrrrrrrrrrrrr}
     	\toprule
        \multicolumn{4}{c}{\multirow{1}{*}{Method}} &  \multicolumn{3}{c}{{Values}} & \multicolumn{3}{c}{{Model size}} \\
        \cmidrule(r){1-4} \cmidrule(r){5-7} \cmidrule(r){8-10}
        index & \makecell{dummy\\variables} & \makecell{stability \\ constraints} & \makecell{stab. cons. \\ merging} & \#opt &  \makecell{time}   &  \makecell{continuous\\relaxation}  & \makecell{number of\\variables}  & \makecell{number of\\constraints} & \makecell{number of\\non-zeros} \\
        \cmidrule(r){1-4} \cmidrule(r){5-7} \cmidrule(r){8-10}
N1 &	&		\eqref{eq:modF7} &		&			3&	144.8&		 \num{747.4}&	 \num{1898}&	 \num{2714}&	 \num{63272}\\
N2 &	x &		\eqref{eq:modF7} &		&			3&	18.7&		 \num{747.4}&	 \num{4146}&	 \num{4146}&	 \num{11278}\\
N3 &	&			\eqref{eq:S1}-\eqref{eq:S5} &	&			3&	10.1&		 \num{744.5}&	 \num{2300}&	 \num{5014}&	 \num{16495}\\
N4 &	x &			\eqref{eq:S1}-\eqref{eq:S5} &	&			3&	9.6&		 \num{744.5}&	 \num{4548}&	 \num{6446}&	 \num{15879}\\
N5 &	&			\eqref{eq:S1}-\eqref{eq:S5} &		\eqref{eq:S3} $\rightarrow$ \eqref{eq:SP3} &		3&	15.9&		 \num{747.8}&	 \num{2300}&	 \num{3465}&	 \num{14946}\\
N6 &	x &			\eqref{eq:S1}-\eqref{eq:S5} &		\eqref{eq:S3} $\rightarrow$ \eqref{eq:SP3} &		3&	12.8&		 \num{747.8}&	 \num{4548}&	 \num{4897}&	 \num{14330}\\
N7 &	&		\eqref{eq:modF7} and \eqref{eq:S1}-\eqref{eq:S5} &		&			3&	16.1&		 \num{744.3}&	 \num{2300}&	 \num{6912}&	 \num{75971}\\
N8 &	x &		\eqref{eq:modF7} and \eqref{eq:S1}-\eqref{eq:S5} &		&			3&	5.4&		 \num{744.3}&	 \num{4548}&	 \num{8345}&	 \num{19676}\\
N9 &	&		\eqref{eq:modF7} and \eqref{eq:S1}-\eqref{eq:S5} &			\eqref{eq:S3} $\rightarrow$ \eqref{eq:SP3} &		3&	25.5&		 \num{746.2}&	 \num{2300}&	 \num{5363}&	 \num{74422}\\
N10 &	x &		\eqref{eq:modF7} and \eqref{eq:S1}-\eqref{eq:S5} &			\eqref{eq:S3} $\rightarrow$ \eqref{eq:SP3} &		3&	9.3&		 \num{746.2}&	 \num{4548}&	 \num{6796}&	 \num{18127}\\
N11 &	&		 \eqref{eq:S1}-\eqref{eq:S5} and \eqref{eq:mix1} &		&			3&	5.6&		 \num{744.3}&	 \num{2300}&	 \num{5311}&	 \num{21377}\\
N12 &	x &		 \eqref{eq:S1}-\eqref{eq:S5} and \eqref{eq:mix1} &		&			3&	3.3&		 \num{744.3}&	 \num{4548}&	 \num{6743}&	 \num{16472}\\
 \bottomrule
\end{tabular}
\end{center}
\end{table}
The table shows a number of interesting facts:
\begin{itemize}
 \item The real-world instances did not challenge the algorithms.
 \item Using dummy variables seems beneficial as algorithms that use them have always a lower average running time than those that do not (e.g., 144.8 seconds for N1 versus 18.7 for N2 and 5.6 seconds for N11 versus 3.3 for N12).
     Even if the decrease in the average number of non-zero elements is less spectacular than for SMTI-GRP, it is still significant for some models (e.g., $\num{63272}$ for N1 versus $\num{11278}$ for N2 and $\num{21277}$ for N11 versus $\num{16472}$ for N12).
\item The kind of stability constraints used by the model appears to have a significant impact on the results: while N1 uses on average 144.8 seconds to solve the SFAS instances, N3 (which uses only the new set of stability constraints) requires merely 10.1 seconds.
    N7, that uses both sets of stability constraints, is a bit worse, even if it has the best continuous relaxation value: 744.3 versus 747.4 for N1 and 744.5 for N3.
    N11, that replaces the original set of stability constraints \eqref{eq:modF7} by \eqref{eq:mix1}, obtains the best relaxation and one of the best average running times.
\item Merging stability constraints \eqref{eq:S3} is not beneficial on these instances, as all algorithms that merge stability constraints have a worse average running compared to those that do not (e.g., 15.9 seconds on average for N5 versus 10.1 for N3, and 25.5 seconds for N9 versus 16.1 for N7).
    This can be explained by the fact that almost no gain is obtained in terms of number of constraints (e.g., 3465 for N5 versus 5014 for N3), but a significant loss is observed in terms of continuous relaxation value (e.g., 747.8 for N5 versus 744.5 for N3).
\item Overall, this computational experiment suggests that the best configurations are N8, N11, and N12.
    All of them use two sets of stability constraints and no constraint merging.
\end{itemize}
\rev{Again, if stability constraints are dropped, the problem becomes solvable in polynomial time \cite{Gab83} and the matching size is increased by at most 1\%.}

Table \ref{table-3a} compares the same twelve combinations on the literature instances SET2.
Overall, the SET2 instances cannot be considered very challenging as all algorithms apart from N1 can solve them to optimality in less than 5 seconds on average.
In addition, we notice that the algorithms' behaviour does not change significantly: (i) dummy variables still seem useful, even if for some configurations (N4 and N12) no significant change is observed, (ii) the kind of stability constraints used still has a significant impact on the overall results, and (iii) using constraint merging still deteriorates the overall results.
We note also that no major difference in the continuous relaxation is observed among the models.
This is probably due to the fact that in 598 instances out of 700, all doctors could be assigned to a hospital, so the continuous relaxation value and the optimal solution were identical.

\begin{table}[bht]
	\scriptsize
	\begin{center}
    	\caption{Comparison of the proposed methods for preprocessed HRT SET2 instances} \label{table-3a}
	\setlength{\tabcolsep}{0.2cm}
	\begin{tabular}{ccccrrrrrrrrrrrrrrr}
     	\toprule
        \multicolumn{4}{c}{\multirow{1}{*}{Method}} &  \multicolumn{3}{c}{{Values}} & \multicolumn{3}{c}{{Model size}} \\
        \cmidrule(r){1-4} \cmidrule(r){5-7} \cmidrule(r){8-10}
        index & \makecell{dummy\\variables} & \makecell{stability \\ constraints} & \makecell{stab. cons. \\ merging} & \#opt &  \makecell{time}   &  \makecell{continuous\\relaxation}  & \makecell{number of\\variables}  & \makecell{number of\\constraints} & \makecell{number of\\non-zeros} \\
        \cmidrule(r){1-4} \cmidrule(r){5-7} \cmidrule(r){8-10}
N1 &	&		\eqref{eq:modF7} &		&			694&	59.7&		 \num{249.9}&	 \num{619}&	 \num{886}&	 \num{16046}\\
N2 &	x &		\eqref{eq:modF7} &		&			700&	1.8&		 \num{249.9}&	 \num{1382}&	 \num{1382}&	 \num{3735}\\
N3 &	&			\eqref{eq:S1}-\eqref{eq:S5} &	&			700&	0.7&		 \num{249.9}&	 \num{781}&	 \num{1667}&	 \num{5315}\\
N4 &	x &			\eqref{eq:S1}-\eqref{eq:S5} &	&			700&	0.8&		 \num{249.9}&	 \num{1544}&	 \num{2163}&	 \num{5297}\\
N5 &	&			\eqref{eq:S1}-\eqref{eq:S5} &		\eqref{eq:S3} $\rightarrow$ \eqref{eq:SP3} &		700&	4.9&		 \num{249.9}&	 \num{781}&	 \num{1193}&	 \num{4841}\\
N6 &	x &			\eqref{eq:S1}-\eqref{eq:S5} &		\eqref{eq:S3} $\rightarrow$ \eqref{eq:SP3} &		700&	1.6&		 \num{249.9}&	 \num{1544}&	 \num{1689}&	 \num{4823}\\
N7 &	&		\eqref{eq:modF7} and \eqref{eq:S1}-\eqref{eq:S5} &		&			700&	1.3&		 \num{249.9}&	 \num{781}&	 \num{2286}&	 \num{20123}\\
N8 &	x &		\eqref{eq:modF7} and \eqref{eq:S1}-\eqref{eq:S5} &		&			700&	0.6&		 \num{249.9}&	 \num{1544}&	 \num{2782}&	 \num{6534}\\
N9 &	&		\eqref{eq:modF7} and \eqref{eq:S1}-\eqref{eq:S5} &			\eqref{eq:S3} $\rightarrow$ \eqref{eq:SP3} &		700&	3.1&		 \num{249.9}&	 \num{781}&	 \num{1812}&	 \num{19649}\\
N10 &	x &		\eqref{eq:modF7} and \eqref{eq:S1}-\eqref{eq:S5} &			\eqref{eq:S3} $\rightarrow$ \eqref{eq:SP3} &		700&	0.8&		 \num{249.9}&	 \num{1544}&	 \num{2308}&	 \num{6060}\\
N11 &	&		 \eqref{eq:S1}-\eqref{eq:S5} and \eqref{eq:mix1} &		&			700&	0.7&		 \num{249.9}&	 \num{781}&	 \num{1794}&	 \num{7706}\\
N12 &	x &		 \eqref{eq:S1}-\eqref{eq:S5} and \eqref{eq:mix1} &		&			700&	0.7&		 \num{249.9}&	 \num{1544}&	 \num{2291}&	 \num{5552}\\
 \bottomrule
\end{tabular}
\end{center}
\end{table}

As all the real-world and literature instances could be solved within an hour, we generated new instances with the instance generator described in \cite{IM09}.
These instances have $759 \times i$ doctors, $53 \times i$ hospitals, and $775 \times i$ available positions, where $i \in \{1, 2, 3, 5, 10\}$ and are called ``RDM$i$" in the following.
Doctors chose between 5 and 6 hospitals, and the tie density on the hospitals' side was equal to 0.85.
For each $i$, 30 instances were created, resulting in 150 instances in total.
These parameters were chosen to mimic the real-world SFAS instances at a larger scale.

Table \ref{table-5} compares the twelve combinations on the RDM$i$ instances. In order to be concise, we only report in the table the number of optimal solutions found and the average running time for each method.

\begin{table}[hbt]
	\scriptsize
	\begin{center}
    	\caption{Comparison of the proposed methods for preprocessed HRT RDM instances} \label{table-5}
	\setlength{\tabcolsep}{0.15cm}
	\begin{tabular}{ccccrrrrrrrrrrrrrrr}
     	\toprule
        \multicolumn{4}{c}{\multirow{1}{*}{Method}} &  \multicolumn{2}{c}{{RDM1}} & \multicolumn{2}{c}{{RDM2}}  & \multicolumn{2}{c}{{RDM3}}  & \multicolumn{2}{c}{{RDM5}} \\
        \cmidrule(r){1-4} \cmidrule(r){5-6} \cmidrule(r){7-8} \cmidrule(r){9-10}  \cmidrule(r){11-12}
        index & \makecell{dummy\\variables} & \makecell{stability \\ constraints} & \makecell{stab. cons. \\ merging} & \#opt &  \makecell{time}   &   \#opt &  \makecell{time}   & \#opt &  \makecell{time} & \#opt &  \makecell{time} \\
        \cmidrule(r){1-4} \cmidrule(r){5-6} \cmidrule(r){7-8} \cmidrule(r){9-10}  \cmidrule(r){11-12}
N1 &	&		\eqref{eq:modF7} &		&			25&	848&	12&	2656&	2&	3485&	0&	3600\\
N2 &	x &		\eqref{eq:modF7} &		&			29&	336&	17&	2056&	9&	2919&	2&	3524\\
N3 &	&			\eqref{eq:S1}-\eqref{eq:S5} &	&			30&	82&	16&	1844&	9&	2667&	2&	3491\\
N4 &	x &			\eqref{eq:S1}-\eqref{eq:S5} &	&			30&	136&	18&	1815&	10&	2694&	2&	3395\\
N5 &	&			\eqref{eq:S1}-\eqref{eq:S5} &		\eqref{eq:S3} $\rightarrow$ \eqref{eq:SP3} &		28&	418&	11&	2358&	4&	3253&	 1&	3574\\
N6 &	x &			\eqref{eq:S1}-\eqref{eq:S5} &		\eqref{eq:S3} $\rightarrow$ \eqref{eq:SP3} &		28&	335&	13&	2253&	6&	3009&	 0&	3600\\
N7 &	&		\eqref{eq:modF7} and \eqref{eq:S1}-\eqref{eq:S5} &		&			30&	66&	25&	1170&	14&	2298&	1&	3495\\
N8 &	x &		\eqref{eq:modF7} and \eqref{eq:S1}-\eqref{eq:S5} &		&			30&	62&	25&	1070&	13&	2267&	4&	3372\\
N9 &	&		\eqref{eq:modF7} and \eqref{eq:S1}-\eqref{eq:S5} &			\eqref{eq:S3} $\rightarrow$ \eqref{eq:SP3} &		30&	134&	 20&	1702&	13&	2763&	1&	3509\\
N10 &	x &		\eqref{eq:modF7} and \eqref{eq:S1}-\eqref{eq:S5} &			\eqref{eq:S3} $\rightarrow$ \eqref{eq:SP3} &		30&	52&	25&	 1023&	14&	2184&	2&	3397\\
N11 &	&		 \eqref{eq:S1}-\eqref{eq:S5} and \eqref{eq:mix1} &		&			30&	46&	24&	1270&	13&	2352&	3&	3434\\
N12 &	x &		 \eqref{eq:S1}-\eqref{eq:S5} and \eqref{eq:mix1} &		&			30&	99&	24&	1252&	14&	2251&	3&	3353\\
 \bottomrule
\end{tabular}
\end{center}
\end{table}

Besides the observations made previously, which are still valid overall, we clearly notice that using two sets of stability constraints is significantly faster, especially for RDM2 and RDM3. In addition, we observe that large instances are extremely difficult, as only four of the RDM5~instances could be solved within an hour of computing time per instance.
We also report that none of the tested algorithms could solve any of the RDM10 instances, even though the models had a reasonable size (e.g., for RDM10, N12 used on average \num{40866} variables, \num{60504}~constraints and \num{145259} non-zero elements).
Finally, even if no ``new" configuration clearly outperforms the others, all of them outperform the literature state-of-the-art model N1.

Table \ref{table-6} studies the impact of the preprocessing developed in \cite{IM09} (we recall that the algorithms given in Section \ref{Pre} are applicable to SMTI instance only), the use of the Gale-Shapley algorithm to give a warm-start to the solver, and the use of priorities on some variables during the branch-and-bound process.
We tested these features on the basic model (N1), and one of the best algorithms (N8).
We selected N8 because it obtained good results on all the datasets we tested in comparison with the other configurations (N8 solved 72 RDM$i$ instances versus 71 for N12 and 70 for N11).

\begin{table}[hbt]
	\scriptsize
	\begin{center}
    	\caption{Comparison of N1 and N8 with additional features on RDM1 and RDM2} \label{table-6}
	\setlength{\tabcolsep}{0.2cm}
	\begin{tabular}{ccccrrrrrrrrrrrrrrr}
     	\toprule
        \multicolumn{4}{c}{\multirow{1}{*}{Method}} &  \multicolumn{5}{c}{{RDM1}} & \multicolumn{5}{c}{{RDM2}} \\
        \cmidrule(r){1-4} \cmidrule(r){5-9} \cmidrule(r){10-14}
        index & \makecell{prep.} & \makecell{warm\\start} & \makecell{variable \\ priorities} & \#opt &  \makecell{time} & \makecell{nb. \\var.}  & \makecell{nb. \\cons.} & \makecell{nb. \\nzs.} & \#opt &  \makecell{time} & \makecell{nb. \\var.}  & \makecell{nb. \\cons.} & \makecell{nb. \\ nzs.} \\
       \cmidrule(r){1-4} \cmidrule(r){5-9} \cmidrule(r){10-14}
N1 &	& 	& 	& 	1&	3581&			 \num{4173}&	 \num{4985}&	 \num{225332}&								0&	3600&			 \num{8343}&	 \num{9967}&	 \num{453286}\\	
N1 &	x &	& 	& 	25&	848&			 \num{1614}&	 \num{2426}&	 \num{37563}&								12&	2656&			 \num{3260}&	 \num{4884}&	 \num{76135}\\	
N1 &	x &	x &	& 28&	524&			 \num{1614}&	 \num{2426}&	 \num{37563}&								15&	2271&			 \num{3260}&	 \num{4884}&	 \num{76135}\\	
N8 &	& 	& 	& 	30&	42&			 \num{9755}&	 \num{18102}&	 \num{43738}&								26&	1036&			 \num{19470}&	 \num{36156}&	 \num{87373}\\	
N8 &	x &	& 	& 	30&	62&			 \num{4075}&	 \num{7304}&	 \num{17026}&								25&	1070&			 \num{8214}&	 \num{14733}&	 \num{34365}\\	
N8 &	x &	x &	& 	30&	82&			 \num{4075}&	 \num{7304}&	 \num{17026}&								25&	957&			 \num{8214}&	 \num{14733}&	 \num{34365}\\	
N8 &	x &	& 	$z_{jk}$ &	30&	37&			 \num{4075}&	 \num{7304}&	 \num{17026}&								26&	840&			 \num{8214}&	 \num{14733}&	 \num{34365}\\	
N8 &	x &	& 	$\YIKT{i}{k}$ and $\YPJKT{j}{k}$ &	30&	82&			 \num{4075}&	 \num{7304}&	 \num{17026}&								 25&	1069&			 \num{8214}&	 \num{14733}&	 \num{34365}\\	
 \bottomrule
\end{tabular}
\end{center}
\end{table}
Unlike for MAX-WT-SMTI-GRP, the preprocessing seems useful for N1 but not for N8. Further investigations showed that, for N8, the inner preprocessing of Gurobi removed a similar amount of variables, constraints, and non-zero elements compared with the preprocessing of~\cite{IM09}.
This was neither the case for N1, nor for M1 and M4 for MAX-WT-SMTI-GRP (with the preprocessing of Section \ref{Pre}).
Using the Gale-Shapley algorithm to provide a warm start allowed N1 to solve an additional six instances from RDM1 and RDM2, but it slightly slowed down N8.
Finally, giving priorities to the $\YIKT{i}{k}$ and $\YPJKT{j}{k}$ variables during the branching process does not seem to help the ILP solver, however, prioritising the $z_{jk}$ variables appears to be beneficial.
Further investigations showed that this statement was true for all the models involving $z_{jk}$~variables (i.e., N3, N4, \dots, N12).

\subsubsection{Master list instances}

As there is no existing set of instances that includes a master list in the literature, we used the generator described in \cite{IM09} to create new data sets.
The same parameters ($759 \times i$ doctors, $53 \times i$ hospitals, and $775 \times i$ available positions, where $i \in \{1, 2, 3, 5, 10\}$) were used, and the grades obtained by the doctors were distributed in $[1,j]$ where $j \in \{5,15,25\}$.
The distribution of doctor grades was controlled using a ``skewedness'' parameter $x$ in the instance generator, which means that the most common doctor score is likely to occur $x$ times more than the least common. Higher values of $x$ hence result in longer ties in the master list, and therefore also in the hospitals' preference lists.
In our experiments we used the value $x=3$.

The constructed instances are called ``RDM-ML-$i$-$j$" in the following.
For each pair $(i,j)$, 30 instances were created, resulting in 450 instances in total.

Table \ref{table-7} compares the literature algorithm N1 and algorithm N8 when priorities are given to the $z_{jk}$ variables (its best configuration). In both cases, preprocessing was applied.

\begin{table}[ht]
	\scriptsize
	\begin{center}
    	\caption{Comparison of N1 and N8 for preprocessed HRT RDM-ML-$i$-$j$ instances} \label{table-7}
	\setlength{\tabcolsep}{0.2cm}
	\begin{tabular}{ccccrrrrrrrrrrrrrrr}
     	\toprule
        \multicolumn{2}{c}{\multirow{1}{*}{Instances}} &  \multicolumn{5}{c}{{N1}} & \multicolumn{5}{c}{{N8 + priorities on $z_{jk}$ variables}} \\
        \cmidrule(r){1-2} \cmidrule(r){3-7} \cmidrule(r){8-12}
        $j$ & $i$ & \#opt &  \makecell{time} & \makecell{nb. \\var.}  & \makecell{nb. \\cons.} & \makecell{nb. \\nzs.} & \#opt &  \makecell{time} & \makecell{nb. \\var.}  & \makecell{nb. \\cons.} & \makecell{nb. \\ nzs.} \\
        \cmidrule(r){1-2} \cmidrule(r){3-7} \cmidrule(r){8-12}
\multirow{5}{*}{25} &	1 &		30&	0.0&	 \num{834}&	 \num{1646}&	 \num{9869}&	30&	0.1&	 \num{2768}&	 \num{4436}&	 \num{9747}\\
&	2 &		30&	0.1&	 \num{1658}&	 \num{3282}&	 \num{19462}&	30&	0.1&	 \num{5528}&	 \num{8843}&	 \num{19420}\\
&	3 &		30&	0.1&	 \num{2464}&	 \num{4900}&	 \num{28656}&	30&	0.2&	 \num{8217}&	 \num{13145}&	 \num{28850}\\
&	5 &		30&	0.2&	 \num{4097}&	 \num{8157}&	 \num{47663}&	30&	0.3&	 \num{13688}&	 \num{21881}&	 \num{48011}\\
&	10 &		30&	0.5&	 \num{8191}&	 \num{16311}&	 \num{95056}&	30&	0.8&	 \num{27356}&	 \num{43738}&	 \num{95971}\\
			  \cmidrule(r){1-2} \cmidrule(r){3-7} \cmidrule(r){8-12}									
\multirow{5}{*}{15} &	1 &		30&	0.1&	 \num{913}&	 \num{1725}&	 \num{12230}&	30&	0.1&	 \num{2738}&	 \num{4565}&	 \num{10159}\\
&	2 &		30&	0.3&	 \num{1804}&	 \num{3428}&	 \num{23945}&	30&	0.2&	 \num{5443}&	 \num{9051}&	 \num{20117}\\
&	3 &		30&	0.5&	 \num{2712}&	 \num{5148}&	 \num{35939}&	30&	0.4&	 \num{8176}&	 \num{13600}&	 \num{30233}\\
&	5 &		30&	3.8&	 \num{4534}&	 \num{8594}&	 \num{60305}&	30&	1.4&	 \num{13647}&	 \num{22715}&	 \num{50516}\\
&	10 &		30&	10.7&	 \num{9000}&	 \num{17120}&	 \num{118923}&	30&	9.5&	 \num{27164}&	 \num{45164}&	 \num{100373}\\
		  \cmidrule(r){1-2} \cmidrule(r){3-7} \cmidrule(r){8-12}						
\multirow{5}{*}{5} &	1 &		29&	227.0&	 \num{1688}&	 \num{2500}&	 \num{49510}&	30&	1.1&	 \num{3868}&	 \num{7244}&	 \num{17052}\\
&	2 &		29&	270.5&	 \num{3336}&	 \num{4960}&	 \num{97526}&	30&	7.0&	 \num{7657}&	 \num{14329}&	 \num{33709}\\
&	3 &		22&	1326.5&	 \num{5098}&	 \num{7534}&	 \num{151203}&	30&	271.3&	 \num{11665}&	 \num{21861}&	 \num{51488}\\
&	5 &		17&	1929.3&	 \num{8297}&	 \num{12357}&	 \num{239442}&	28&	521.7&	 \num{19047}&	 \num{35641}&	 \num{83829}\\
&	10 &		6&	3034.2&	 \num{16963}&	 \num{25083}&	 \num{494286}&	21&	1734.3&	 \num{38807}&	 \num{72733}&	 \num{171301}\\
 \bottomrule
\end{tabular}
\end{center}
\end{table}
It appears that instances that have a master list are significantly easier than those that do not, as each method can solve at least one instance of each group, even for RDM-ML-10-$j$.
In addition, it seems that datasets allowing a larger range for the grades (RDM-ML-$i$-25 and RDM-ML-$i$-15), are easier as even the basic algorithm N1 can solve them all in seconds. Difficult master list instances have a very narrow range of grades (e.g., $j=5$), and many doctors and hospitals (e.g., $i=5$ or $10$).
On these instances, we can appreciate the benefits of the proposed improvements, as N1 can only solve 6 RDM-ML-10-5 instances, while N8 with priorities can solve 21 of them.
Overall, for the 90 RDM-ML-10-$j$ instances generated, N1 can solve 66, whereas N8 with priorities can solve 81, an increase of 23\%.

\subsection{Summary of the experiments}  \label{s4}

We empirically showed that, overall, we could solve significantly larger instances of MAX-WT-SMTI-GRP, MAX-SMTI, and MAX-HRT when compared to the existing methods.
However, we observed that each problem had its own peculiarities.

For our practical case of MAX-WT-SMTI-GRP (the Coram application), characterised by a medium number of agents and very long preference lists, it is of paramount importance to reduce the number of non-zero elements and, thus, it is crucial to use dummy variables.
Indeed, when preference lists are very long, stability constraints involve many variables and the size of the models increases quickly.
To a lesser extent, it is beneficial to decrease the number of constraints without deteriorating the continuous relaxation too much.
Thus, using constraint merging and double stability constraints is useful. In conclusion, configurations M4 and~M6 are the most suitable for the problem.

For our MAX-SMTI instances, characterised by a very large number of agents and shorter preference lists, it is vital to reduce the number of constraints and, thus, it is advised to use constraint merging.
Indeed, under these conditions, the models involve many stability constraints, which can be difficult to tackle by ILP solvers, even if they do not involve as many variables as they did for MAX-WT-SMTI-GRP.
To a lesser extent, it is beneficial to decrease the number of non-zero elements, but only when it is not at the expense of a significant increase in terms of variables and constraints.
Thus, using dummy variables is beneficial when the tie density is high and when the preference list length is not too small.
For this problem, configurations M3 and M4 are the most suitable.

For our MAX-HRT instances with no master list, characterised by a medium number of agents and short preference lists, it is important to help the solver reduce the gap between the lower and upper bounds and, thus, it is advised to use two sets of stability constraints.
Indeed, we observed that the solver struggled to solve some instances even when the size of the model was reasonable.
For this problem, configurations N8 and N12 are the most suitable.

The MAX-HRT instances with a master list that we tested, which aimed to mimic the real-world SFAS instances at a larger scale, did not present any sort of challenge when the grade range was reasonably large (i.e., at least 15 in our experiments). For instances with a very narrow grade range (i.e., 5 in our experiments), using N8 with preprocessing and the adequate priorities is advised.

Finally, we observed that preprocessing was useful for MAX-WT-SMTI-GRP, but not always useful for MAX-HRT (e.g., with configuration N8) or for MAX-SMTI.
In addition, we saw that using a warm start could help some algorithms (e.g., N1 for MAX-HRT), and slow down some others (e.g., N8 for MAX-HRT).
Finally, we also empirically observed that giving priorities to some specific sets of variables during the branching process of the ILP solver could be beneficial (e.g., giving priority to the $z_{jk}$ variables of N8 for MAX-HRT).

\section{Conclusion}\label{Con}
We described two algorithms for preprocessing instances of SMTI where ties occur on both sides.
This resulted in significant improvements when applied to models from the literature, solving an additional 15 (of 22 total) real-world MAX-WT-SMTI-GRP instances from Coram within one hour per instance.
We also introduced new ILP models for SMTI, first by using dummy variables to reduce the complexity of the constraint matrix, and then by merging stability constraints, and using double stability constraints.
Various combinations of techniques were demonstrated to improve the performance of our models, and together with the earlier preprocessing our new models solved all 22 Coram instances with a mean runtime of less than one minute.
Computational experiments on randomly generated instances also showed that our models could solve instances of SMTI with up to $\num{50000}$ agents per side.
The new ILP models were also extended to HRT, where we showed a performance improvement from 144 seconds to 3 seconds on average on real-world instances from SFAS.
We also showed that we could solve an additional $23\%$ of randomly generated instances with around \num{7500} doctors and hospital places when compared to state-of-the-art models.

\rev{In this paper we have not considered issues of strategy, however this direction is certainly worthy of further study in instances of SMTI-GRP and HRT.  In the context of the classical Hospitals / Residents problem, it is a well-known result that, with respect to the Resident-oriented Gale-Shapley algorithm \cite{GS62,GI89}, it is a dominant strategy for the doctors to tell the truth \cite{Rot85a}.  On the other hand, there is no mechanism that is stable and strategy-proof for hospitals \cite{Rot86}.  Relative to any mechanism in SMTI-GRP or HRT that is based on finding a maximum cardinality stable matching, it is not difficult to show that weights or preferences (respectively) could be falsified by doctors and/or hospitals in order to improve their outcomes relative to their true preferences (e.g., by declaring less desirable preferences as unacceptable).  We leave as future work the investigation of the existence of strategy-proof mechanisms for instances of SMTI-GRP and HRT that produce good approximations to maximum cardinality stable matchings.}

Further work \rev{also} includes extending our preprocessing algorithms for SMTI to the more general case of HRT.
This may lead to the exact solution of larger HRT instances than we considered in this paper.
It also remains open to extend our models to the extension of HRT where couples apply jointly to pairs of hospitals.

\section*{Acknowledgements}
\rev{We would like to thank the two anonymous reviewers for their valuable comments that have helped to improve the presentation of this paper.}  We would also like to thank Kevin Yong of Coram for the collaboration involving the pairing of children with families, and also Rob Irving for helpful advice regarding the usage of the ``skewedness'' parameter that was used when generating HRT instances with master lists.  This research was supported by the Engineering and Physical Science Research Council through grant EP/P029825/1 (first four authors) and grant EP/P028306/1 (fifth and sixth authors).

\bibliography{matching}
\bibliographystyle{plain}

\end{document}